\documentclass[11pt,draftcls,onecolumn]{IEEEtran}

\ifCLASSINFOpdf
\else
   \usepackage[dvips]{graphicx}
\fi
\usepackage{url}
\usepackage{graphicx}
\usepackage{amsmath}
\usepackage{comment}
\usepackage[utf8]{inputenc}
\usepackage[english]{babel}
\usepackage{color}
\usepackage{framed}
\usepackage{cite}
\usepackage{amsmath,amssymb,amsfonts}
\usepackage{graphicx}
\usepackage{textcomp}
\usepackage{cite}
\usepackage{bbold}
\usepackage{ifpdf}
\usepackage{times}
\usepackage{layout}
\usepackage{float}
\usepackage{afterpage}
\usepackage{amsmath}
\usepackage{amstext}
\usepackage{amssymb,bm}
\usepackage{latexsym}
\usepackage{color}
\usepackage{flexisym}
\usepackage{dblfloatfix}    
\usepackage{kantlipsum}
\usepackage{subcaption}
\usepackage{graphicx}
\usepackage{amsmath}
\usepackage{amsthm}
\usepackage{graphicx}
\usepackage{caption}
\usepackage{subcaption}
\usepackage{booktabs}
\usepackage{multicol}
\usepackage{lipsum}

\usepackage{enumitem}
\usepackage{amsmath}
\usepackage{algpseudocode}
\usepackage{amsthm}
\usepackage{dsfont}
\usepackage{algorithm}
\usepackage{thmtools}
\usepackage{amssymb}
\usepackage[dvipsnames]{xcolor}

\def\bp{\boldsymbol{p}}

\def\br{\boldsymbol{r}}

\def\bz{\boldsymbol{z}}

\def\bB{\boldsymbol{B}}

\def\bI{\boldsymbol{I}}

\def\bP{\boldsymbol{P}}

\def\bR{\boldsymbol{R}}

\def\bmu{\boldsymbol{\mu}}

\def\bSigma{\boldsymbol{\Sigma}}
\def\btau{\boldsymbol{\tau}}

\newtheorem{theorem}{Theorem}

\newtheorem{lemma}{Lemma}

\theoremstyle{definition}

\algnewcommand\algorithmicinput{\textbf{Input:}}
\algnewcommand\Input{\item[\algorithmicinput]}
\algnewcommand\algorithmicoutput{\textbf{Output:}}
\algnewcommand\Output{\item[\algorithmicoutput]}
\algnewcommand\algorithmicinit{\textbf{Initialize:}}
\algnewcommand\Init{\item[\algorithmicinit]}

\makeatletter
\newcommand*{\rom}[1]{\expandafter\@slowromancap\romannumeral #1@}
\makeatother

\begin{document}

\title{Covariance Recovery for One-Bit Sampled Data
With Time-Varying Sampling Thresholds--- \\ \noindent Part~\rom{2}: Non-Stationary Signals}

\author{Arian Eamaz, \IEEEmembership{Student Member, IEEE}, Farhang Yeganegi, and \\ Mojtaba Soltanalian, \IEEEmembership{Senior Member, IEEE}
\thanks{This work was supported in part by the National Science Foundation Grants CCF-1704401 and ECCS-1809225. Parts of this work were presented at the International Conference on Acoustics, Speech and Signal Processing
(ICASSP) 2021, held in Toronto, Canada \cite{eamaz2021modified}. The first two authors contributed equally to this work.}
\thanks{A. Eamaz, F. Yeganegi and M. Soltanalian are with the Dept. of Electrical and Computer Engineering, University of Illinois Chicago, Chicago, USA (e-mails: aeamaz2@uic.edu, fyegan2@uic.edu).}
}

\markboth{IEEE TRANSACTIONS ON SIGNAL PROCESSING , 2022}
{Shell \MakeLowercase{\textit{et al.}}: Bare Demo of IEEEtran.cls for IEEE Journals}
\maketitle

\begin{abstract}
The recovery of the input signal covariance values from its one-bit sampled counterpart has been deemed a challenging task in the literature. To deal with its difficulties, some assumptions are typically made to find a relation between the input covariance matrix and the autocorrelation values of the one-bit sampled data. This includes the arcsine law and the modified arcsine law that were discussed in Part~\rom{1} of this work \cite{eamazpart1}. We showed that by facilitating the deployment of time-varying thresholds, the modified arcsine law has a promising performance in covariance recovery. However, the modified arcsine law also assumes input signals are stationary, which is typically a simplifying assumption for real-world applications. In fact, in many signal processing applications, the input signals are readily known to be non-stationary with a non-Toeplitz covariance matrix. In this paper, we propose an approach to extending the arcsine law to the case where one-bit ADCs apply time-varying thresholds while dealing with input signals that originate from a non-stationary process. In particular, the recovery methods are shown to accurately recover the time-varying variance and autocorrelation values. Furthermore, we extend the formulation of the Bussgang law to the case where non-stationary input signals are considered.
\end{abstract}


\IEEEpeerreviewmaketitle

\section{Introduction}
\IEEEPARstart{C}{ovariance} matrix recovery plays an important role in statistical signal processing applications such as directions of arrival (DOA) estimation, radar waveform design, target parameter estimation, communication channel estimation, and adaptive radar detection  \cite{maronna1976robust,weiss2020blind,bose2021efficient,cheng2017mimo,bekkerman2006target,xu2008target,djelouat2021joint,mahot2013asymptotic}. When digital signal processing is concerned, using one-bit quantization and digitization, in which the input signals are compared with given threshold levels, allows for sampling at a very high rate and with lower energy consumption \cite{instrumentsanalog,mezghani2018blind,ameri2018one,sedighi2020one}. As a result of employing one-bit sampling, however,  we can only use the sign data as partial available information to recover the signal covariance, and  second order statistics in general, making it more challenging. In \cite{van1966spectrum,jacovitti1987methods,jacovitti1994estimation,bar2002doa}, the authors have considered the input signal as a stationary zero-mean Gaussian process, and with this assumption, the input covariance is recovered by taking advantage of the \emph{arcsine law} which connects the covariance of an unquantized signal with that of its quantized counterpart \cite{liu2017one,ameri2018one}. In \cite{bussgang1952crosscorrelation}, the relationship between the cross-correlation matrix of the input signal and the one-bit output data is characterized as the \emph{Bussgang law} for the stationary zero-mean Gaussian signals. Note that the sampling threshold levels are considered to be zero in these research efforts. The zero threshold values can give rise to some difficulties in signal amplitude recovery and a considerable portion of signal information may be lost. As a natural alternative, time-varying thresholds are utilized in recent works which can lead to enhancements in recovery performance \cite{eamaz2021modified,qian2017admm,gianelli2016one,khobahi2020model,khobahi2018signal,wang2017angular,xi2020gridless}.\par
Owing to the successful performance of time-varying sampling thresholds for signals amplitude recovery, such time-varying thresholds were considered for the covariance recovery problem \cite{eamaz2021modified}, and more extensively in Part~\rom{1} of this work \cite{eamazpart1}, exhibiting a significantly improved performance in the estimation of signal autocorrelation values via a \emph{modified arcsine law}.
Moreover, taking time-varying thresholds into consideration for cross-correlation matrix recovery, promising results were demonstrated with a \emph{modified Bussgang law} in Part~\rom{1} of this work \cite{eamazpart1}.\par
A critical restriction of the arcsine law, as well as the modified arcsine law, lies in the necessary assumption of a stationary input signal\cite{van1966spectrum,eamaz2021modified,eamazpart1}. In real-world communication and digital signal processing applications, however, input signals are non-stationary in general and have time-varying variances \cite{haykin1995nonlinear,hu2017enhanced,grenier1983time}. In such scenarios, covarince recovery is an even more prominent tool in the analysis of non-stationary processes and systems, and can provide useful insights into their innate dynamics \cite{tsai2020non,engle2019large,fox2015bayesian}.
Nevertheless, in a non-stationary environment, the expected accuracy of covariance recovery is typically diminished.  

In this paper, we present an approach to extend our modified arcsine law for time-varying sampling thresholds, discussed in Part~\rom{1} of this work \cite{eamazpart1}, to recover signal covariance matrices with an arbitrary non-Toeplitz structure. Moreover, a Bussgang law with time-varying thresholds is established for the non-stationary scenario.

\subsection{Contribution of the Paper}
We study the covariance recovery for a non-stationary input signals in one-bit quantization systems deploying time-varying thresholds. In particular, we formulate an integral-based relation between the autocorrelation function of the one-bit sampled data and the generic covariance matrix entries of the input signal. Moreover, a closed-form formulation for the mean of the input signal is obtained and the utilized to recover the time-varying signal variances. It is demonstrated that to recover the autocorrelation values, we should evaluate the obtained integral which appears to be intractable analytically. To approach this problem, we first employ a one-point piece-wise \text{Padé} approximation (PA) to recast the integrands as rational expressions which are readily integrable. Next, we formulate an estimation criterion to recover the desired input autocorrelation values. The accuracy of the PA is also investigated. In the next step, two well-known numerical integration techniques are employed to estimate the input autocorrelation values; namely, the Gauss-Legendre quadrature and the Monte-Carlo integration techniques. Interestingly, the proposed estimation criteria for these approaches take convex form, facilitating an accelerated recovery. Lastly, a modified Bussgang law for non-stationary input signals is presented. By using the modified Bussgang law, the matrix elements associated with the cross-correlation between the input signal and the one-bit sampled data can be recovered. Several numerical results are presented to illustrate the effectiveness of the proposed methodologies.
\subsection{Organization of the Paper}
Section~\rom{2} is dedicated to formulating the autocorrelation function of the one-bit sampled data with time-varying thresholds in the case of non-stationary inputs. In Section~\rom{3}, the time-varying variances are recovered by using the proposed formula for the mean of the one-bit sampled data. Sections~\rom{4} presents our \text{Padé} Approximation (PA) to recover the input signal autocorrelation sequence. Subsequently,
Sections~\rom{5} and \rom{6} discuss two widely-known numerical integration techniques, i.e. the Gauss-Legendre quadrature and the Monte-Carlo integration methods, applied to our arbitrary non-Toeplitz covariance matrix recovery problem. Section~\rom{7} is where the various methods proposed for covariance recovery are compared. A proper thresholding for covariance recovery through the estimation of the threshold mean is discussed in Section~\rom{8}. The modified Bussgang law for time-varying thresholds in the case of 
non-stationary signals is presented in Section~\rom{9}. Finally, Section~\rom{10} concludes the paper.

\vspace{5pt}

\underline{\emph{Notation:}}
We use bold lowercase letters for vectors, and bold uppercase letters for matrices and uppercase letters for matrix entries. For instance, $\bR_{\mathbf{x}}$ and $R_{\mathbf{x}}(i,j)$ denote the autocorrelation matrix and the $ij$-th element of the autocorrelation matrix of the vector $\mathbf{x}$, respectively. $(\cdot)^{\top}$ and $(\cdot)^{\mathrm{H}}$ denote the vector/matrix transpose, and the Hermitian transpose, respectively. $\mathbb{E}\left\{.\right\}$ denotes the expected value. The Frobenius norm of a matrix $\bB\in \mathbb{C}^{M\times N}$ is defined as $\|\bB\|_{\mathrm{F}}=\sqrt{\sum^{M}_{r=1}\sum^{N}_{s=1}\left|B(r,s)\right|^{2}}$ where $\{B(r,s)\}$ are entries of $\bB$. For an event $\mathcal{E}$, $\mathbb{I}_{(\mathcal{E})}$ is the indicator function for that event; i.e. $\mathbb{I}_{(\mathcal{E})}$ is $1$ if $\mathcal{E}$ occurs and $0$ otherwise. The $Q$-function is defined as
\begin{equation}
\label{eq:135}
\begin{aligned}
Q(x) &= \frac{1}{\sqrt{2 \pi}} \int_{x}^{\infty} \exp \left(-\frac{z^{2}}{2}\right) \,dz.
\end{aligned}
\end{equation}
Further, $Q^{-1}$ is an inverse $Q$-function. The error function (erf) is defined as $
\operatorname{erf} x=\frac{2}{\sqrt{\pi}} \int_{0}^{x} e^{-z^{2}}\,d z$.
The incomplete Gamma function is defined as
\begin{equation}
\label{eq:136}
\Gamma(s, x)=\int_{x}^{\infty} z^{s-1} e^{-z}\,dz.
\end{equation}
Finally, the cumulative distribution function (CDF) of a zero mean Gaussian process $\bz\sim\mathcal{N}(0,\zeta)$ is defined as
\begin{equation}
\label{eq:1bbb}
\Psi(\bz) = \frac{1}{\sqrt{2\pi}}\int^{z}_{-\infty}e^{-\frac{t^{2}}{2\zeta^{2}}} \,dt.
\end{equation}
\section{Covariance Recovery In Non-Stationary Scenario}
We assume that the input signal is a zero-mean non-stationary Gaussian process $\mathbf{x}\sim\mathcal{N}\left(\mathbf{0},\bR_{\mathbf{x}}\right)$, where $\bR_{\mathbf{x}}$ is the non-Toeplitz covariance matrix of $\mathbf{x}$ with the time-varying diagonal elements. Specifically, the input signal is supposed to be non wide sense stationary (WSS) or weak stationary (for abbreviation we use "non-stationary" instead of "non-WSS"). A signal is wide sense stationary (weak stationary) if and only if:
\begin{itemize}
    \item[a.] The signal mean is constant over time.
    \item[b.] $R_{\mathbf{x}}(i,j) = R_{\mathbf{x}}(l)$. This means that we have a constant variance over time; i.e. $R_{\mathbf{x}}(i,i)=R_{\mathbf{x}}(0)$.
    \item[c.] $\mathbb{E}\left\{x^2_{i}\right\}<\infty$ or generally $\mathbb{E}\left\{x^{2n}_{i}\right\}<\infty$ \cite{bollerslev1994arch,kendall1987kendall,hayes2009statistical}.
\end{itemize}
In our case, the input signal can satisfy (a) and (c), however, the autocorrelation function does not rely on lag $R_{\mathbf{x}}(i,j) \neq R_{\mathbf{x}}(l)$, therefore, $\bR_{\mathbf{x}}=\left[R_{\mathbf{x}}(i,j)\right]$ is a non-Toeplitz covariance matrix with distinct diagonal elements (time-varying variance). Consequently, our signal is non-stationary or marginal heteroskedastic \cite{bollerslev1994arch,hayes2009statistical}. A simple but famous example for such signals are those originating in \emph{Wiener processes} or \emph{Brownian motion} $\eta_{t}\sim \mathcal{N}\left(0,t\right)$. The input signal $\mathbf{x}\in \mathbb{R}^{N}$ is considered to be an arbitrary temporal sequence of the distribution ensembles $\left\{\mathbf{x}(k)\right\}$ whith $k \in \left\{1,\cdots,N_{\mathbf{x}}\right\}$.
\subsection{Modified Arcsine Law For Non-Stationary Input Signals}
We consider a non-zero time-varying Gaussian threshold~$\btau$ that is independent of the input signal with the distribution $\btau\sim\mathcal{N}\left(\mathbf{d}=\mathbf{1}d,\bSigma\right)$, and define a new random process $\mathbf{w}$ such that $\mathbf{w}=\mathbf{x}-\btau$. Clearly, $\mathbf{w}$ is a non-stationary Gaussian process with $\mathbf{w}\sim\mathcal{N}\left(-\mathbf{d},\bR_{x}+\bSigma=\bP\right)$ where $\bP$ is a non-Toeplitz matrix. The autocorrelation function of the one-bit quantized output process is formulated in the following.
\begin{theorem}
\label{theorem_2}
Suppose $p_{0i}=\bP(i,i)$, $p_{0j}=\bP(j,j)$ and $p_{ij}=\bP(i,j)$, where $\bP$ is the covariance matrix of $\mathbf{w}$. Consider the one-bit quantized random variable $\mathbf{y}=f(\mathbf{w})$ where $f(.)$ is the sign function. Then, the autocorrelation function of $\mathbf{y}$ takes the form
\begin{equation}
\label{eq:114}
\begin{aligned}
R_{\mathbf{y}}(i,j)=\frac{e^{\frac{-d^2(p_{0i}+p_{0j}-2p_{ij})}{2(p_{0i}p_{0j}-p_{ij}^2)}}}{\pi\sqrt{\left(p_{0i}p_{0j}-p_{ij}^{2}\right)}}\left\{ \int_{0}^{\frac{\pi}{2}} \frac{1}{\beta_{n}}+\sqrt{\frac{\pi}{\beta_{n}}} \frac{\alpha_{n}}{2\beta_{n}} e^{\frac{\alpha_{n}^{2}}{4 \beta_{n}}}\right.\\\left.-\sqrt{\frac{\pi}{\beta_{n}}} \frac{\alpha_{n}}{\beta_{n}} Q\left(\frac{\alpha_{n}}{\sqrt{2 \beta_{n}}}\right) e^{\frac{\alpha_{n}^{2}}{4 \beta_{n}}} d \theta\right\}-1,
\end{aligned}
\end{equation}
where $\alpha_{n}$ and $\beta_{n}$ are evaluated as
\begin{equation}
\label{eq:115}
\begin{aligned}
\alpha_{n} &= \frac{d\left(p_{0i}\sin\theta+p_{0j}\cos\theta-p_{ij}(\cos\theta+\sin\theta)\right)}{(p_{0i}p_{0j}-p_{ij}^2)},\\
\beta_{n} &= \frac{p_{0j}\cos^2\theta+p_{0i}\sin^2\theta-p_{ij}\sin 2\theta}{2(p_{0i}p_{0j}-p_{ij}^2)}.
\end{aligned}
\end{equation}
\end{theorem}

\begin{proof}
The autocorrelation value associated with lags $i$ and $j$ is given by
\begin{equation}
\label{eq:27}
R_{\mathbf{y}}(i,j) = \kappa \int_{-\infty}^{\infty} \int_{-\infty}^{\infty}\hspace{-.1cm} f(w_i)f(w_j)e^{\lambda(d)} dw_i dw_j
\end{equation}
where $\kappa$ and $\lambda(d)$ are defined as
\begin{equation}
\label{eq:86}
\kappa \triangleq \left(2\pi\sqrt{p_{0i}p_{0j}-p_{ij}^2}\right)^{-1},
\end{equation}
\begin{equation}
\label{eq:28}
\lambda(d)\hspace{-.1cm}\triangleq\frac{(w_i+d)^2p_{0j}+(w_j+d)^2p_{0i}-2p_{ij}(w_i+d)(w_j+d)}{-2(p_{0i}p_{0j}-p_{ij}^2)}.
\end{equation}
The autocorrelation function in (\ref{eq:27}) can be rewritten as
\begin{equation}
\label{eq:29}
\begin{aligned}
R_{\mathbf{y}}(i,j) = & \kappa \left(\int_{0}^{\infty} \int_{0}^{\infty} e^{\lambda(d)} \,dw_i \,dw_j\right.\\
&+\int_{-\infty}^{0} \int_{-\infty}^{0} e^{\lambda(d)} \,dw_i \,dw_j \\
&-\int_{0}^{\infty} \int_{-\infty}^{0} e^{\lambda(d)} \,dw_i \,dw_j \\
&\left.-\int_{-\infty}^{0} \int_{0}^{\infty} e^{\lambda(d)} \,dw_i \,dw_j\right).
\end{aligned}
\end{equation}
We can simplify (\ref{eq:29}) using the relation $\kappa\int_{-\infty}^{\infty}\int_{-\infty}^{\infty} e^{\lambda(d)} \,dw_i\,dw_j=1$.
In fact, one can verify that
\begin{equation}
\label{eq:31}
R_{\mathbf{y}}(i,j) = 2\kappa \int_{0}^{\infty}\hspace{-.1cm} \int_{0}^{\infty} \hspace{-.1cm}\left(e^{\lambda(d)}+e^{\lambda(-d)}\right) \,dw_i\,dw_j -1.
\end{equation}
By employing polar coordinates $w_i=\rho \cos \theta$, $w_j=\rho \sin \theta$, we can recast the integral in (\ref{eq:31}) as
\begin{equation}
\label{eq:32}
\begin{aligned}
R_{\mathbf{y}}(i,j) = \chi \int_{0}^{\frac{\pi}{2}} \hspace{-.1cm} \int_{0}^{\infty} \hspace{-.1cm} e^{-\beta_{n}\rho^2} \hspace{-.1cm}\left(e^{-\alpha_{n}\rho}+e^{\alpha_{n}\rho}\right)\rho\,d\rho\,d\theta -1,
\end{aligned}
\end{equation}
where $\alpha_{n}$ and $\beta_{n}$ are readily defined in (\ref{eq:115}), and
\begin{equation}
\label{eq:79}
\chi \triangleq 2\kappa e^{\frac{-d^2(p_{0i}+p_{0j}-2p_{ij})}{2(p_{0i}p_{0j}-p_{ij}^2)}}.
\end{equation}
Integrating (\ref{eq:32}) with respect to $\rho$ leads to
\begin{equation}
\label{eq:35}
\begin{aligned}
R_{\mathbf{y}}(i,j)=\frac{e^{\frac{-d^2(p_{0i}+p_{0j}-2p_{ij})}{2(p_{0i}p_{0j}-p_{ij}^2)}}}{\pi\sqrt{\left(p_{0i}p_{0j}-p_{ij}^{2}\right)}}\left\{ \int_{0}^{\frac{\pi}{2}} \frac{1}{\beta_{n}}+\sqrt{\frac{\pi}{\beta_{n}}} \frac{\alpha_{n}}{2\beta_{n}} e^{\frac{\alpha_{n}^{2}}{4 \beta_{n}}}\right.\\\left.-\sqrt{\frac{\pi}{\beta_{n}}} \frac{\alpha_{n}}{\beta_{n}} Q\left(\frac{\alpha_{n}}{\sqrt{2 \beta_{n}}}\right) e^{\frac{\alpha_{n}^{2}}{4 \beta_{n}}} d \theta\right\}-1.
\end{aligned}
\end{equation}
\end{proof}

It remains to evaluate the integral in (\ref{eq:114}) in terms of $\{p_{0i}\}$, $\{p_{0j}\}$ and $\{p_{ij}\}$, which have to be estimated---a task that is central to our efforts in the rest of this paper. Finding $\{p_{0i}\}$, $\{p_{0j}\}$ and $\{p_{ij}\}$ results in time-varying input variance and autocorrelation recovery, which can be achieved by considering the relation:
\begin{equation}
\label{eq:36}
\bR_{\mathbf{x}}(i,j) = \bP(i,j)-\bSigma(i,j).
\end{equation}
For $i=j$, the input variance is hence given by $\bR_{\mathbf{x}}(i,i) = r_{0i} = p_{0i}-\bSigma(i,i)$ and $\bR_{\mathbf{x}}(j,j) = r_{0j} = p_{0j}-\bSigma(j,j)$, while for $i\neq j$, we have the input autocorrelation as $\bR_{\mathbf{x}}(i,j) = r_{ij} = p_{ij}-\bSigma(i,j)$.
\section{Time-Varying Variance Recovery}
\label{tvv}
To recover the time-varying variances $\{r_{0i}\}$, the following lemma would be useful.
\begin{lemma}
\label{lemma_1}
The first moment (mean) of the one-bit sampled data, depends on the threshold distribution and the power of sampled data via the relation 
\begin{equation}
\label{rem1}
 \mathbb{E}\left\{y_{i}\right\} = 2Q\left(\frac{d}{\sqrt{p_{0i}}}\right)-1, \quad \forall i\in \{1,\cdots,N\}.
\end{equation}
\begin{proof}
We have
\begin{equation}
\label{pr1-1}
\mathbb{E}\left\{y_{i}\right\} = \int_{-\infty}^{+\infty} f(w_{i}) p(w_{i}) \,dw_{i}
\end{equation}
for $i\in \{1,\cdots,N\}$, with $p(w_{i})=\left(\sqrt{2\pi p_{0i}}\right)^{-1} e^{\frac{-\left(w_{i}+d\right)^{2}}{2p_{0i}}}$. We can further simplify (\ref{pr1-1}) as
\begin{equation}
\label{pr1-2}
\begin{aligned}
\mathbb{E}\left\{y_{i}\right\} &= -\int_{-\infty}^{0} p(w_{i}) \,dw_{i}+\int_{0}^{\infty} p(w_{i}) \,dw_{i}\\ &= 2\int_{0}^{+\infty} p(w_{i}) \,dw_{i}-1\\ &=2Q\left(\frac{d}{\sqrt{p_{0i}}}\right)-1
\end{aligned}
\end{equation}
which completes the proof.
\end{proof}
\end{lemma}
In light of the above, a relation between the input variance and the mean of one-bit sampled data is established, which provides an additional avenue to estimate the variances $\left\{p_{0i}\right\}$. More precisely,  
according to Lemma~\ref{lemma_1} and (\ref{eq:36}), the input time-varying variances $\left\{r_{0i}\right\}$ are given by
\begin{equation}
\label{eq:fast_1}
\begin{aligned}
r_{0i}^{\star}&=\left(\frac{d}{Q^{-1}\left(\frac{\mu_{i}+1}{2}\right)}\right)^{2}-\sigma^{2}_{\btau}, \quad i\in \left\{1,\cdots,N\right\},
\end{aligned}
\end{equation}
where $\{r_{0i}^{\star}\}$  denote the optimal values of $\{r_{0i}\}$, $\sigma^{2}_{\btau}$ is the threshold variance, and $\{\mu_{i}\}$ denote the entries of $\bmu$ which may be estimated via the sample mean $\frac{1}{N_{\mathbf{x}}}\sum^{N_{\mathbf{x}}}_{k=1}\mathbf{y}(k)$ \cite{kendall1987kendall}.

\begin{figure*}[t]
	\centering
	\begin{subfigure}[b]{0.45\textwidth}
		\includegraphics[width=1\linewidth]{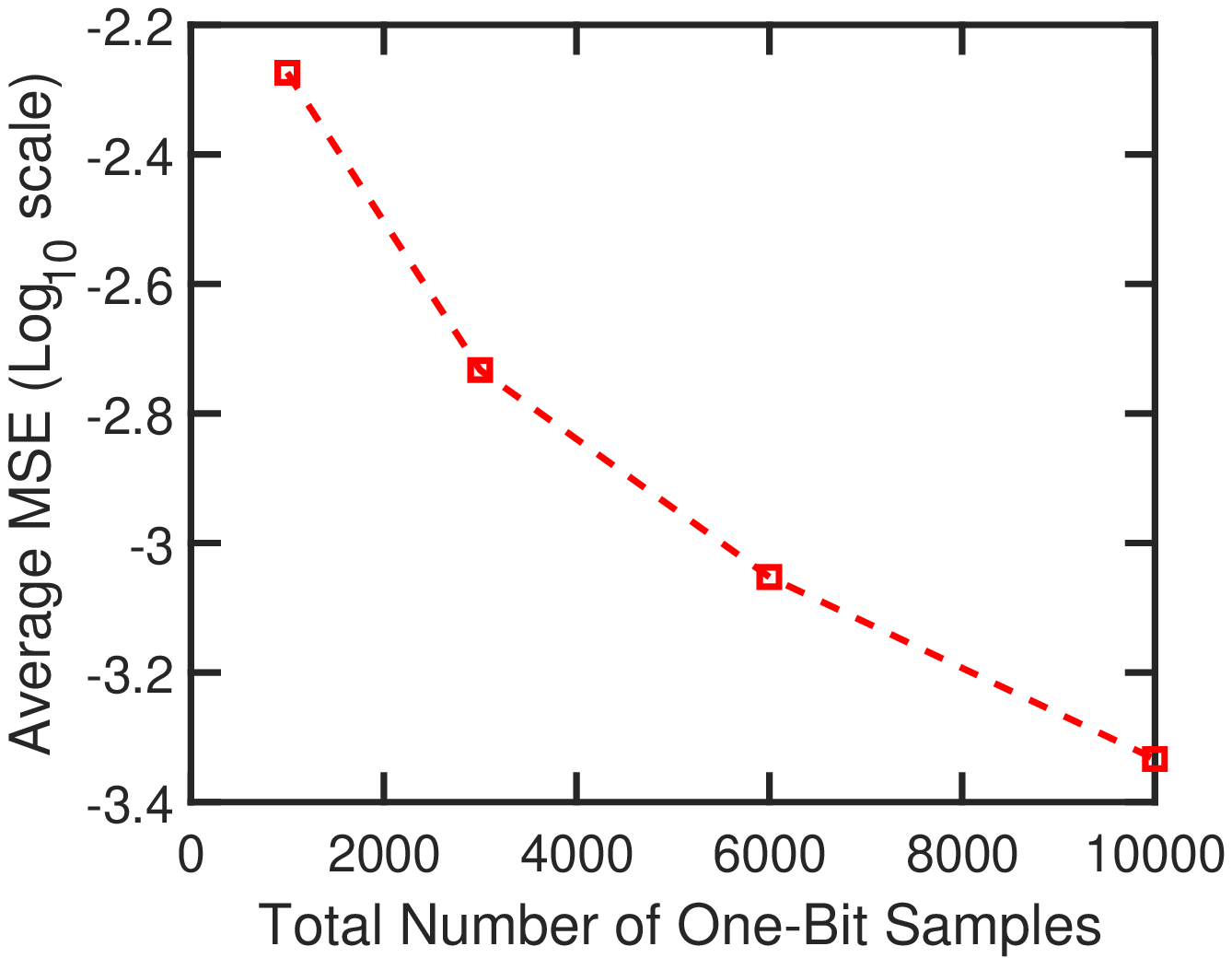}
		\caption{}
	\end{subfigure}
	\begin{subfigure}[b]{0.45\textwidth}
		\includegraphics[width=1\linewidth]{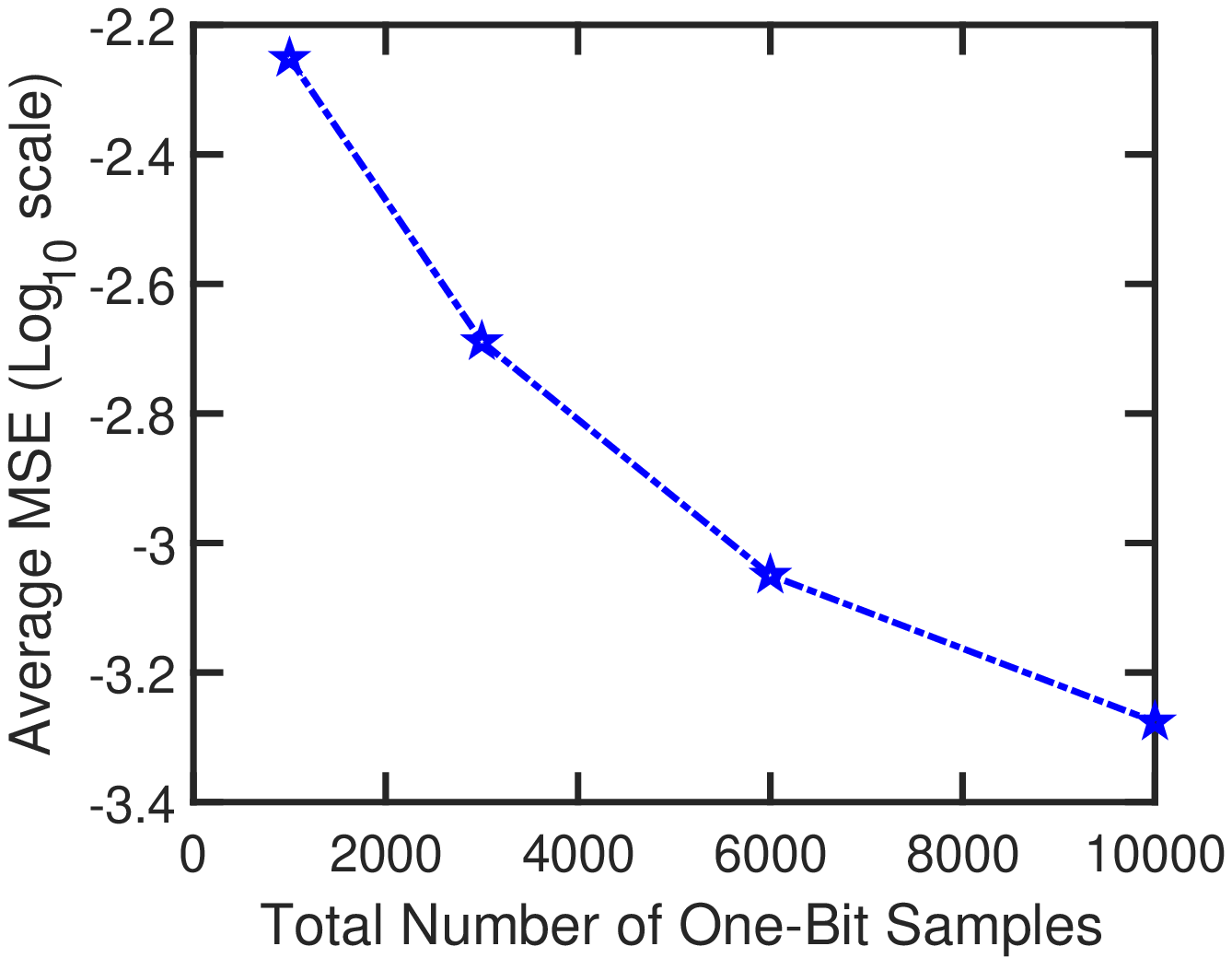}
		\caption{}
	\end{subfigure}
	\caption{Average MSE for time-varying input variance recovery with (a) $i=2$ ($r_{02}$), and (b) $j=8$ ($r_{08}$), for different one-bit sample sizes.}
	\label{figure_5}
\end{figure*}

\subsection{Numerical Results}
We will examine the effectiveness of (\ref{eq:fast_1}) in estimating the time-varying input variances. In all experiments, the input signals were generated as zero-mean Wiener process sequences with time-varying variance ranging from $0.2$ to $0.8$. Also, the number of states is $100$ (i.e., $\mathbf{x}\in \mathbb{R}^{100}$). Accordingly, we made use of the time-varying thresholds with $d=0.5$ and diagonal $\bSigma$ whose diagonal entries are set to $0.2$. In the non-stationary input signal case, the true input variance is not a constant number. Therefore, we define the experimental mean square error (MSE) of the estimate $\hat r_{0i}$ of a variance $r_{0i}$ as
\begin{equation}
\label{eq:46}
\begin{aligned}
\mathrm{MSE} = \frac{1}{E} \sum^{E}_{e=1} |r^{e}_{0i}-{\hat r^{e}_{0i}}|^{2},
\end{aligned}
\end{equation}
where $\left\{r^{e}_{0i},\hat r^{e}_{0i}\right\}$ are the time-varying variances and their estimates in the $e$-th experiment. Also, the number of experiments is assumed to be $E=15$.
As can be seen in Fig.~\ref{figure_5}, we can accurately estimate the time-varying variance elements of an input signal based on (\ref{eq:fast_1}) for $i=2$ and $j=8$ ($r_{02} \neq r_{08}$). The results are obtained for the number of ensembles $N_{\mathbf{x}}\in \left\{ 1000, 3000, 6000, 10000\right\}$, with fixed $d$ and $\bSigma$ for each experiment. It is observed that the accuracy of time-varying variance recovery will significantly improve as the number of one-bit samples grows large.

To further investigate the effectiveness of our proposed approach, we generate a non-stationary Gaussian process according to $\mathbf{x}\sim\mathcal{N}\left(0,\sigma^{2}_{t}\bI\right)$, where $\{\sigma^{2}_{t}\}$ are generated based on the generalized autoregressive conditional heteroskedasticity (GARCH) model with order one, i.e. $\operatorname{GARCH}\left(1,1\right)$, which may be written as \cite{bollerslev1994arch},
\begin{equation}
 \label{garch}
 \sigma^{2}_{t} = \zeta_{0}+\zeta_{1}\sigma^{2}_{t-1}+\zeta_{2}\epsilon^{2}_{t-1}, \quad x_{t} = \epsilon_{t}|\psi_{t-1},
\end{equation}
where $\{x_{t}\}$ are elements of $\mathbf{x}$, $\epsilon_{t}|\psi_{t-1}$ denotes the conditional random variable $\epsilon_{t}$ given its previous ensembles set $\psi_{t-1}$, and $\left\{\zeta_{0},\zeta_{1},\zeta_{2}\right\}$ are our GARCH model parameters. In Fig.~\ref{variance_recovery}, we present an example of time-varying variance sequence recovery. The true input signal time-varying variance $\sigma^{2}_{t}$ and the estimated values by our approach are presented when $t$ is a temporal sequence of length $20$.

So far, we have obtained the time-varying variance elements of the input covariance matrix. To recover the input autocorrelation values ($\{r_{ij}\}, i\neq j $), the integral in (\ref{eq:114}) should be evaluated which appears to be difficult to find in closed-form. Therefore, in the following, we deploy various approximations to facilitate its evaluation, which enables to the recovery of all elements of the covariance matrix $\bR_{\mathbf{x}}$.
\begin{figure}[t]
	\center{\includegraphics[width=0.6\textwidth]{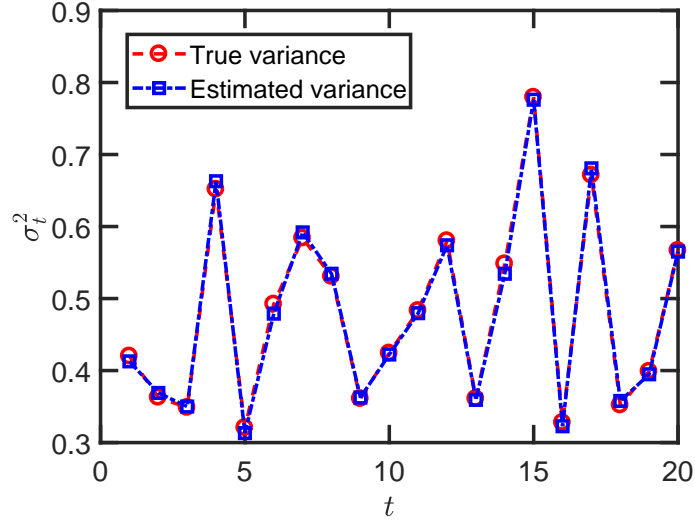}}
	\caption{Recovery of the input signal time-varying variance $\sigma^{2}_{t}$ generated by a $\operatorname{GARCH}\left(1,1\right)$ model based on (\ref{eq:fast_1}), when $t$ is a temporal sequence of length $20$, with the true values plotted along the estimates.}
	\label{variance_recovery}
\end{figure}

\section{Analytic Approach for Covariance Recovery}
\label{sec:analytic}
The first part of the integration in (\ref{eq:114}) can be analytically evaluated as
\begin{equation}
\label{eq:38}
\begin{aligned}
\int_{0}^{\frac{\pi}{2}} \frac{1}{\beta_{n}} \,d\theta &= \sqrt{p_{0i}p_{0j}-p_{ij}^2} \left(\pi+2\sin^{-1}\left[\frac{p_{ij}}{\sqrt{p_{0i}p_{0j}}}\right]\right).
\end{aligned}
\end{equation}
If the threshold is considered as a zero-mean Gaussian process ($\btau\sim \mathcal{N}\left(0,\bSigma\right)$), one can resort to the well-known  \emph{arcsine law} relation for non-stationary signals, i.e.,  $R_{\mathbf{y}}(i,j)=\frac{2}{\pi}\sin^{-1}\left(\frac{p_{ij}}{\sqrt{p_{0i}p_{0j}}}\right)$. However, in the general case, we evaluate all parts of the integration in (\ref{eq:114}). Computing the integral in (\ref{eq:114}) with the integrands $D_{1}\left(\theta;p_{0i},p_{0j},p_{ij},d\right)=\sqrt{\frac{\pi}{\beta_{n}}} \frac{\alpha_{n}}{\beta_{n}} Q\left(\frac{\alpha_{n}}{\sqrt{2 \beta_{n}}}\right) e^{\frac{\alpha_{n}^{2}}{4 \beta_{n}}}$ and $D_{2}\left(\theta;p_{0i},p_{0j},p_{ij},d\right)  =\sqrt{\frac{\pi}{\beta_{n}}} \frac{\alpha_{n}}{2\beta_{n}} e^{\frac{\alpha_{n}^{2}}{4 \beta_{n}}}$ with respect to $\theta$ appears to be a difficult task. Thus, in the following section, the \text{Padé} approximation (PA) \cite{basdevant1972pade,brezinski1995taste,gonnet2013robust} is utilized to approximate said integrands $D_{1}$ and $D_{2}$. This facilitates the recovery of $\{p_{ij}\}$ in Section~\ref{subsec:4}.

\subsection{Padé Approximation}
\label{subsec:3}
As in Part \rom{1} of this work \cite{eamazpart1}, we use PA to approximate $D_{1}$ and $D_{2}$.  Note that the integration in (\ref{eq:114}) occurs in the interval $\theta \in \left[0,\frac{\pi}{2}\right]$. To have a better fitness, we again use the idea of piece-wise PA with three distinct intervals $\left[0,\frac{\pi}{8}\right]$, $\left[\frac{\pi}{8},\frac{3\pi}{8}\right]$, and $\left[\frac{3\pi}{8},\frac{\pi}{2}\right]$ with the expansion points $\theta\in\left\{0,\frac{\pi}{4},\frac{\pi}{2}\right\}$. Consequently, the function $D_{2}\left(\theta;p_{0i},p_{0j},p_{ij},d\right)$ is approximated as
\begin{equation}
\label{eq:37}
\begin{aligned}
\theta &\in \left[0,\frac{\pi}{8}\right] \cup \left[\frac{\pi}{8},\frac{3\pi}{8}\right] \cup \left[\frac{3\pi}{8},\frac{\pi}{2}\right]:\\ D_{2}&=\sqrt{\frac{\pi}{\beta_{n}}} \frac{\alpha_{n}}{2\beta_{n}} e^{\frac{\alpha_{n}^{2}}{4 \beta_{n}}}\approx \frac{e+s\theta}{k+g\theta+h\theta^2},
\end{aligned}
\end{equation}
A similar approximation can be proposed for $D_{1}\left(\theta;p_{0i},p_{0j},p_{ij},d\right)$. It is straightforward to verify that the two functions $D_{1}\left(\theta;p_{0i},p_{0j},p_{ij},d\right)$ and $D_{2}\left(\theta;p_{0i},p_{0j},p_{ij},d\right)$ are analytic at the expansion points. Also, the $Q$-function in (\ref{eq:114}) is approximated by the $\bar{Q}$-function as \cite{chiani2003new}:
\begin{equation}
\label{eq:11}
\bar{Q}\left(x\right) = \frac{1}{12} e^{\frac{-x^2}{2}} +\frac{1}{4} e^{\frac{-2x^2}{3}}, \quad x > 0.
\end{equation}

Substituting $D_{2}\left(\theta;p_{0i},p_{0j},p_{ij},d\right)$ with its approximation and evaluating the integration in the associated parts of (\ref{eq:114}) results in:
\begin{equation}
\label{eq:39}
\begin{aligned}
\int_{0}^{\frac{\pi}{8}}& \sqrt{\frac{\pi}{\beta_{n}}}\frac{\alpha_{n}}{2\beta_{n}}e^{\frac{\alpha_{n}^2}{4\beta_{n}}} \,d\theta \approx \frac{s}{2h}\ln{\left(\frac{\left|k+\frac{\pi g}{8}+\frac{\pi^2h}{64}\right|}{\left|k\right|}\right)}+\\&\frac{2eh-sg}{h\sqrt{4hk-g^2}}\tan^{-1}\left(\frac{\pi h\sqrt{4hk-g^2}}{16hk+\pi gh}\right),
\end{aligned}
\end{equation}
\begin{equation}
\label{eq:40}
\begin{aligned}
\int_{\frac{\pi}{8}}^{\frac{3\pi}{8}}& \sqrt{\frac{\pi}{\beta_{n}}}\frac{\alpha_{n}}{2\beta_{n}}e^{\frac{\alpha_{n}^2}{4\beta_{n}}} \,d\theta \approx \frac{s}{2h} \ln\left(\frac{\left|64k+9\pi ^2 h+24\pi hg\right|}{\left|64k+\pi ^2 h+8\pi hg\right|}\right)+\\&\frac{2eh-sg}{h\sqrt{4kh-g^2}}\tan^{-1}\left(\frac{8\pi h\sqrt{4hk-g^2}}{64kh+3\pi^2h^2+16\pi hg}\right),
\end{aligned}
\end{equation}
\begin{equation}
\label{eq:41}
\begin{aligned}
\int_{\frac{3\pi}{8}}^{\frac{\pi}{2}}& \sqrt{\frac{\pi}{\beta_{n}}}\frac{\alpha_{n}}{2\beta_{n}}e^{\frac{\alpha_{n}^2}{4\beta_{n}}} \,d\theta \approx \frac{s}{2h} \ln\left(\frac{\left|k+\frac{\pi g}{2}+\frac{\pi^2h}{4}\right|}{\left|k+\frac{3\pi g}{8}+\frac{9\pi^2 h}{64}\right|}\right)+\\&\frac{2eh-sg}{h\sqrt{4kh-g^2}}\tan^{-1}\left(\frac{\pi h\sqrt{4hk-g^2}}{16kh+3\pi^2h^2+7\pi hg}\right).
\end{aligned}
\end{equation}
Similar approximations can be obtained for terms associated with the function $D_{1}\left(\theta;p_{0i},p_{0j},p_{ij},d\right)$.


\subsection{Recovery Criterion}
\label{subsec:4}
Based on our discussions in Section~\ref{tvv}, $p_{0i}^{\star}$ and $p_{0j}^{\star}$ may be immediately obtained by (\ref{eq:fast_1}). Then, $\{p_{ij}\}$ are estimated by formulating a minimization problem. For this purpose, one may consider the following criterion:
\begin{equation}
\label{eq:42}
\begin{aligned}
&\bar{G}(p_{0i},p_{0j},p_{ij}) \triangleq \log\left(\left|R_{\mathbf{y}}(i,j)-\chi \left\{ \int_{0}^{\frac{\pi}{2}} \frac{1}{\beta_{n}} \right.\right.\right.\\& \hspace{-.3cm} \left.\left.\left.+\sqrt{\frac{\pi}{\beta_{n}}} \frac{\alpha_{n}}{2\beta_{n}} e^{\frac{\alpha_{n}^{2}}{4 \beta_{n}}}-\sqrt{\frac{\pi}{\beta_{n}}} \frac{\alpha_{n}}{\beta_{n}} Q\left(\frac{\alpha_{n}}{\sqrt{2 \beta_{n}}}\right) e^{\frac{\alpha_{n}^{2}}{4 \beta_{n}}} d \theta\right\}+1\right|^2\right),
\end{aligned}
\end{equation}
where the autocorrelation of output signal ($R_{\mathbf{y}}$) can be estimated via the sample covariance matrix (SCM) \cite{hayes2009statistical},
\begin{equation}
\label{eq:43}
\begin{aligned}
\bR_{\mathbf{y}}\approx \frac{1}{N_{\mathbf{x}}} \sum_{k=1}^{N_{\mathbf{x}}} \mathbf{y}(k) \mathbf{y}(k)^{\mathrm{H}},
\end{aligned}
\end{equation}
with $\{\mathbf{y}(k)\}$ being the observed sign vectors, and $\chi$ being the same as defined in (\ref{eq:79}). Note that by now we have derived an approximated version of~(\ref{eq:114}) using PA. Let $H_{n}(p_{0i},p_{0j},p_{ij})$ denote this approximation. Therefore, we can alternatively consider the criterion:
\begin{equation}
\label{eq:44}
\begin{aligned}
G(p_{ij}) &\triangleq \log\left(\left|R_{\mathbf{y}}(i,j)-H_{n}(p_{0i}^{\star},p_{0j}^{\star},p_{ij})\right|^2\right).
\end{aligned}
\end{equation}
A numerical investigation of \eqref{eq:44} reveals that it is multi-modal, i.e. with multiple local minima---see Fig.~\ref{figure_7} for an example of the optimization landscape of $G(p_{ij})$. Taking the feasible region of $\left\{p_{ij}\right\}$ into account, we can formulate the recovery problem:
\begin{equation}
\label{eq:45}
\begin{aligned}
\mathcal{P}_{i,j}&: &\min_{p_{ij}}& &G(p_{ij}),& &\mbox{s.t.}& &-p_{m}\leq p_{ij} \leq p_{m},
\end{aligned}
\end{equation}
where $p_{m}=\operatorname{min}\{[p_{0i}^{\star},p_{0j}^{\star}]\}$. The non-convex problem in (\ref{eq:45}) may be solved via the gradient descent numerical optimization approach by employing multiple random initial points. Once $p_{ij}$ is estimated, we can estimate the autocorrelation values of $\mathbf{x}$ via (\ref{eq:36}). The acquired optimum recovery results will be presented in the following.

\begin{figure}[t]
	\center{\includegraphics[width=0.6\textwidth]{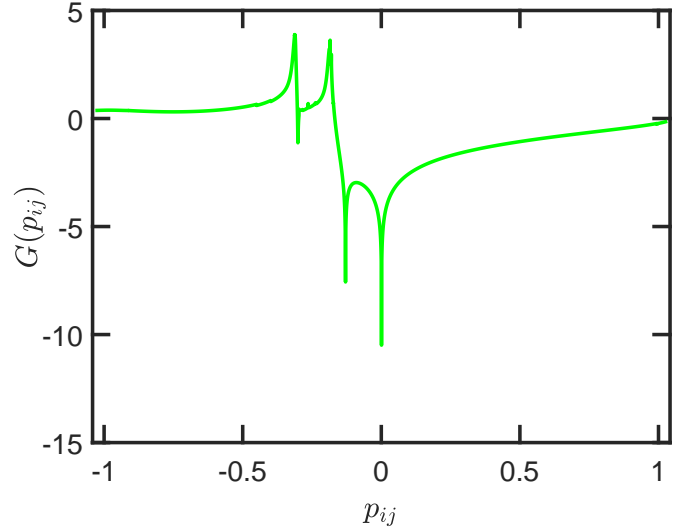}}
	\caption{Example plot of the estimation criterion $G(p_{ij})$ with respect to $p_{ij}$ showing its multi-modality, i.e. having multiple local optima.}
	\label{figure_7}
\end{figure}

\subsection{Numerical Results}
\label{subsec:22}
We will examine the effectiveness of the PA method by comparing its recovery results with the true input signal autocorrelation values in the non-stationary case. In all experiments, the input signals were generated as zero-mean Wiener process sequences with time-varying variance ranging from $0.2$ to $0.8$. The number of states is set to $100$ (i.e., $\mathbf{x}\in \mathbb{R}^{100}$). Accordingly, we make use of the time-varying thresholds with $d=0.5$ and diagonal $\bSigma$ whose diagonal entries are set to $0.2$.

\begin{figure}[t]
	\center{\includegraphics[width=0.6\textwidth]{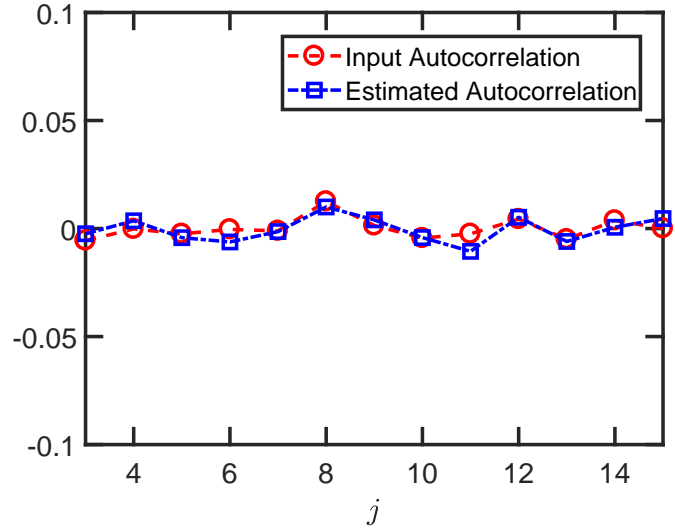}}
	\caption{Recovery of the input signal autocorrelation $r_{ij}$ using the PA approach for $i=2$ with $j$ being a temporal sequence of length $13$, with the true values plotted along the estimates.}
	\label{figure_4}
\end{figure}

We first present an example of autocorrelation sequence recovery. The true input signal autocorrelation and the estimated autocorrelation values by our approach are shown in Fig.~\ref{figure_4}, where $i=2$ and $j$ is a temporal sequence of length $13$. Fig.~\ref{figure_4} appears to confirm the possibility of recovering the autocorrelation values from one-bit sampled data with time-varying thresholds in the non-stationary case.

\subsection{Fitness Analysis of the Proposed Approximations}
\label{sec_fitness}
In this section, we further examine the capability of the PA approach to approximate the integrands in  (\ref{eq:114}). In fact, it appears a precise approximation of the integrands is  connected to the one-bit comparison thresholds that are used for sampling. To see how, note that the exponential term in  (\ref{eq:114}), i.e. $e^{\frac{\alpha_{n}^{2}}{4\beta_{n}}}$, should remain bounded to guarantee a well-behaved \text{Padé} approximation. This is due to the fact that, in (\ref{eq:114}), $\beta_{n}$ is non-zero and $e^{\frac{\alpha_{n}^{2}}{4\beta_{n}}}$ is the only term that can grow very fast and create a round-off numerical error. Consider a bounding of this term in the form
\begin{equation}
\label{eq:131}
e^{\frac{\alpha_{n}^{2}}{4\beta_{n}}}<\gamma_{1}.
\end{equation}
Since $\ln(\cdot)$ is a strictly increasing function, Eq.~(\ref{eq:131}) can be alternatively expressed as:
\begin{equation}
\label{eq:132}
\frac{\alpha_{n}^{2}}{4\beta_{n}}<\ln(\gamma_{1}).
\end{equation}
Note that $\alpha_{n}$ in (\ref{eq:115}) is \emph{directly scaled by the sampling threshold mean $d$}. By setting $\alpha_{n} = d \alpha_{e}$, (\ref{eq:132}) can be written as
\begin{equation}
\label{eq:133}
d^{2}\left(\frac{\alpha_{e}^{2}}{4\beta_{n}}\right) < \ln(\gamma_{1}),
\end{equation}
where $\alpha_{n}$ and $\beta_{n}$ are defined in (\ref{eq:115}).
To guarantee the bound in (\ref{eq:133}), we should have:
\begin{equation}
\label{eq:134}
d^{2} \max_{\theta} \left\{ \frac{\alpha_{e}^{2}}{4\beta_{n}} \mathbb{I}_{\left(0\leq \theta \leq \frac{\pi}{2}\right)}\right\} < \ln(\gamma_{1}).
\end{equation}
The inner optimization problem $\max_{\theta} \frac{\alpha_{e}^{2}}{4\beta_{n}}\mathbb{I}_{\left(0\leq \theta \leq \frac{\pi}{2}\right)}$ can be solved via golden section search and parabolic interpolation method \cite{sun2006optimization}. Note that $\gamma_{1}$ is greater than one. To prove this claim, it is sufficient to show
\begin{equation}
\label{eq:134n}
\exists ~\theta \in \left[0,\frac{\pi}{2}\right]: \quad \beta_{n}(\theta)>0,
\end{equation}
since $\alpha_{e}^{2}$ is always non-negative. By plugging in $\theta=0$ in $\beta_{n}$, we obtain $\beta_{n}\bigg|_{\theta=0}=0.5p_{0j}\left(p_{0i}p_{0j}-p_{ij}^{2}\right)^{-1}$ which is always positive.

To assess the goodness of the considered approximations, consider the integrand term:
\begin{equation}
\label{eq:125}
\begin{aligned}
\Delta(\theta; p_{0i},p_{0j},p_{ij}) &= \sqrt{\frac{\pi}{\beta_{n}}} \frac{\alpha_{n}}{2\beta_{n}} e^{\frac{\alpha_{n}^{2}}{4 \beta_{n}}}\\&-\sqrt{\frac{\pi}{\beta_{n}}} \frac{\alpha_{n}}{\beta_{n}} \bar{Q}\left(\frac{\alpha_{n}}{\sqrt{2 \beta_{n}}}\right) e^{\frac{\alpha_{n}^{2}}{4 \beta_{n}}}.
\end{aligned}
\end{equation}
We compare $\Delta(\cdot)$ with its approximated forms for $\theta \in \left[0,\frac{\pi}{8}\right] \cup \left[\frac{\pi}{8},\frac{3\pi}{8}\right] \cup \left[\frac{3\pi}{8},\frac{\pi}{2}\right]$, whose results are plotted in Fig.~\ref{figure_32} with parameters $p_{0i}=0.8$, $p_{0j}=0.7$, $p_{ij}=0.05$, and $d=0.7$. Note that, in this case, (\ref{eq:134}) is satisfied by considering $\gamma_{1}=2$. As can be seen, PA appears to promise good fitness, with a small mean square error (MSE) of $\sim 1.1699e-04$.

\begin{figure}[t]
	\center{\includegraphics[width=0.6\textwidth]
		{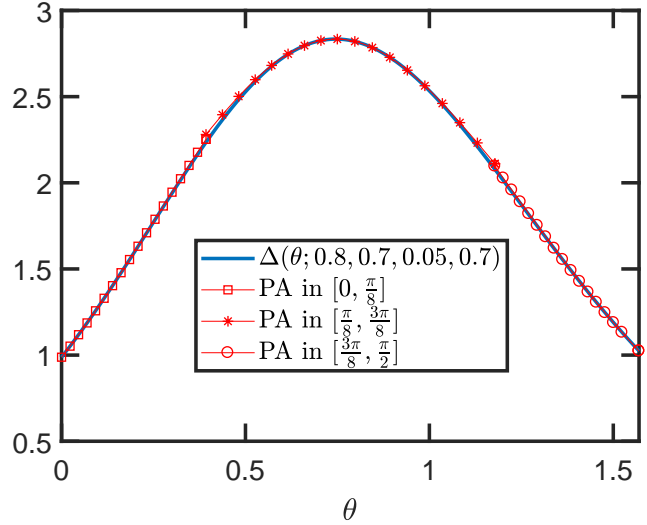}}
	\caption{Comparing $\Delta$ with its approximated piece-wise function using PA.}
	\label{figure_32}
\end{figure}

\section{Gaussian Quadrature Technique for Covariance Recovery}
In this section, the Gauss-Legendre quadrature approach is adopted to evaluate the integration in (\ref{eq:114}). This lays the ground for the recovery of $\left\{p_{ij}\right\}$ since $p_{0i}$ and $p_{0j}$ are obtained by (\ref{eq:fast_1}). At first, an approximated version of (\ref{eq:114}) is obtained based on the Gauss-Legendre quadrature technique in Section \ref{subsec:7}. Then, a criterion will be presented to recover $\left\{p_{ij}\right\}$, and subsequently, the input autocorrelation values in Section \ref{subsec:8}. Finally, the efficacy of this approach in estimating the input autocorrelation values is numerically evaluated.

\subsection{Gauss-Legendre Quadrature Method for Integral Approximation}
\label{subsec:7}
As discussed in Part~\rom{1} of this work, the central assumption to the use of the Gauss-Legendre quadrature technique is that the integrand $f(x)$ should be finite within the domain of integration. The integrands in (\ref{eq:114}) meet this assumption; it is easy to verify that $\textbf{num}(\beta_{n})\neq 0$, where $\textbf{num}(\cdot)$ denotes the numerator of the fractional argument. Therefore, by employing the Gauss-Legendre quadrature technique, the relation in (\ref{eq:114}) can be approximated as
\begin{equation}
\label{eq:94}
\begin{aligned}
R_{\mathbf{y}}(i,j) &\approx \frac{e^{\frac{-d^2(p_{0i}+p_{0j}-2p_{ij})}{2(p_{0i}p_{0j}-p_{ij}^2)}}}{\pi\sqrt{\left(p_{0i}p_{0j}-p_{ij}^{2}\right)}}\left\{\int_{0}^{\frac{\pi}{2}} \frac{1}{\beta_{n}} d \theta \right.\\ &\left. -\frac{\pi}{4}\sum_{\varepsilon=1}^{N_{q}} \omega_{\varepsilon} D_{1}\left(\frac{\pi}{4}(\theta_{\varepsilon}+1); p_{0i}, p_{0j}, p_{ij}, d\right) \right. \\ &\left. +\frac{\pi}{4}\sum_{\varepsilon=1}^{N_{q}} \omega_{\varepsilon} D_{2}\left(\frac{\pi}{4}(\theta_{\varepsilon}+1); p_{0i}, p_{0j}, p_{ij}, d\right)\right\}-1,
\end{aligned}
\end{equation}
where $\theta_{\varepsilon}$ denotes the $\varepsilon$-th Gauss node. Note that the first part of the above integration was readily given in closed-form in~(\ref{eq:38}).

\begin{figure}[t]
	\center{\includegraphics[width=0.6\textwidth]{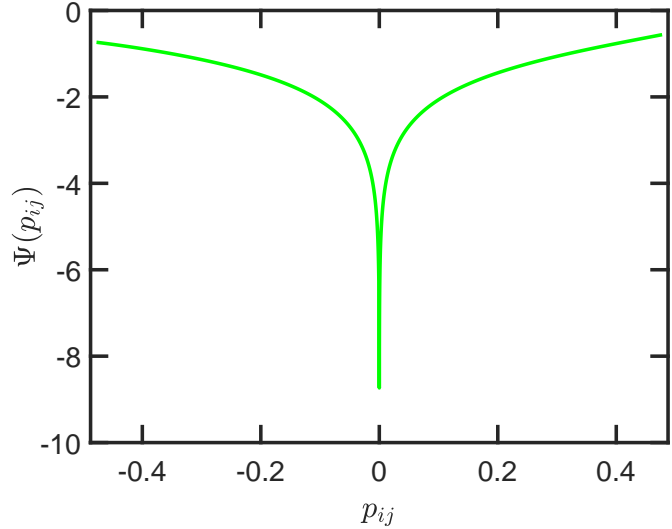}}
	\caption{Example plot of the Gauss-Legendre quadrature approach-based estimation criterion $\Psi(p_{ij})$ with respect to $p_{ij}$ showing its convexity.}
	\label{figure_13}
\end{figure}

\subsection{Covariance Recovery via Convex Optimization}
\label{subsec:8}
Based on our discussion in Section~\ref{tvv}, the values for $p_{0i}^{\star}$ and $p_{0j}^{\star}$ are simply given by (\ref{eq:fast_1}). The parameters of interest $\left\{p_{ij}\right\}$ are then estimated by formulating a minimization problem; namely, we consider the following criterion:
\begin{equation}
\label{eq:95}
\begin{aligned}
&\bar{\Psi}(p_{ij}) \triangleq \log\left(\left|R_{\mathbf{y}}(i,j)-\chi\left\{\int_{0}^{\frac{\pi}{2}}\frac{1}{\beta_{n}} d\theta \right.\right.\right.\\ &\left.\left.\left.-\frac{\pi}{4}\sum_{\varepsilon=1}^{N_{q}} \omega_{\varepsilon} D_{1}\left(\frac{\pi}{4}(\theta_{\varepsilon}+1); p_{0i}^{\star}, p_{0j}^{\star}, p_{ij}, d\right)\right.\right.\right.\\ &\left.\left.\left.+\frac{\pi}{4}\sum_{\varepsilon=1}^{N_{q}} \omega_{\varepsilon} D_{2}\left(\frac{\pi}{4}(\theta_{\varepsilon}+1); p_{0i}^{\star}, p_{0j}^{\star}, p_{ij}, d\right)\right\}+1\right|^{2}\right),
\end{aligned}
\end{equation}
for which the autocorrelation of output signal $R_{\mathbf{y}}$ can be estimated using the SCM in (\ref{eq:43}), and $\chi$ is the same as that in (\ref{eq:79}). Recall that we have obtained an approximated version of (\ref{eq:114}) using the Gauss-Legendre quadrature in (\ref{eq:94}). Let $J_{n}(p_{ij})$ denote this approximation. As a result, we can alternatively use the criterion:
\begin{equation}
\label{eq:96}
\Psi(p_{ij}) \triangleq \log\left(\left|R_{\mathbf{y}}(i,j)-J_{n}(p_{ij})\right|^{2}\right).
\end{equation}
It is interesting to note that the criterion in (\ref{eq:96}) is a convex function with respect to $p_{ij}$ (a proof is provided in Appendix~A)---see Fig.~\ref{figure_13} for an example of the optimization landscape of $\Psi(p_{ij})$. By considering the feasible region of $p_{ij}$, the following recovery problem is obtained:
\begin{equation}
\label{eq:97}
\begin{aligned}
\mathcal{P}_{i,j}&: &\min_{p_{ij}}& &\Psi(p_{ij}),& &\mbox{s.t.}& &-p_{m}\leq p_{ij} \leq p_{m},
\end{aligned}
\end{equation}
where $p_{m}$ is defined in Section~\ref{subsec:4}. The convex problem in (\ref{eq:97}) may be solved efficiently using the golden section search and parabolic interpolation approach. Once $\left\{p_{ij}\right\}$ is obtained, one can estimate the autocorrelation values of $\mathbf{x}$ via (\ref{eq:36}). The recovery results will be presented in the following.


\subsection{Numerical Results}
\label{subsec:30}
We examine the usefulness of the Gauss-Legendre quadrature technique by comparing its recovery results with the true input signal autocorrelation values in the non-stationary case. In all experiments, the input signals were generated as zero-mean Gaussian sequences with time-varying variance ranging from $0.2$ to $0.8$. Accordingly, we made use of the time-varying thresholds with $d=0.3$ and diagonal $\bSigma$ whose diagonal entries are equal to $0.1$.

We present an example of autocorrelation sequence recovery. The true input signal autocorrelation and our estimated autocorrelation values are shown in Fig.~\ref{figure_40} with $i=2$ and $j$ being a temporal sequence of length $13$. Fig.~\ref{figure_40} appears not only to confirm the possibility of recovering the autocorrelation values from one-bit sampled data with time-varying thresholds in the non-stationary case but also the effectiveness of the Gauss-Legendre technique.
\begin{figure}[t]
	\center{\includegraphics[width=0.6\textwidth]{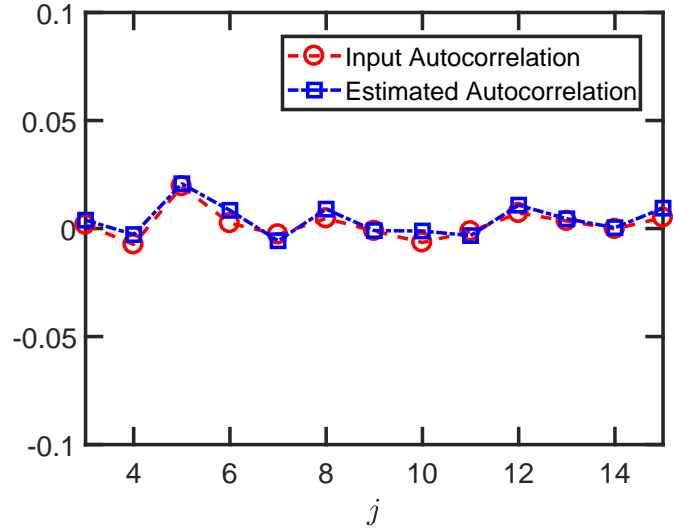}}
	\caption{Recovery of the input autocorrelation $r_{ij}$ using the Gauss-Legendre technique, with $i=2$ and $j$ being a temporal sequence of length $13$. The true values are plotted alongside the estimates.}
	\label{figure_40}
\end{figure}
\section{Monte-Carlo Integration for Covariance Recovery}
In this section, the Monte-Carlo integration approach is utilized to evaluate the integral in (\ref{eq:114}). At first, we formulate an approximated version of (\ref{eq:114}) based on the Monte-Carlo integration approach in Section \ref{subsec:11}. We then present a new criterion to recover $\{p_{ij}\}$ based on this approximation. The efficacy of this approach in estimating the input autocorrelation values is numerically evaluated.

\subsection{Monte-Carlo Method for Integral Evaluation}
\label{subsec:11}
As mentioned in Part~\rom{1} of this work, the Monte-Carlo integration technique can be utilized to recover the input signal covariance matrix. Here, the same idea is adopted to evaluate the integration in (\ref{eq:114}) for the non-stationary input signal scenario. More concretely, by employing the Monte-Carlo integration technique, one can approximate (\ref{eq:114}) as follows:
\begin{equation}
\label{eq:103}
\begin{aligned}
R_{\mathbf{y}}(i,j) &\approx \frac{e^{\frac{-d^2(p_{0i}+p_{0j}-2p_{ij})}{2(p_{0i}p_{0j}-p_{ij}^2)}}}{\pi\sqrt{\left(p_{0i}p_{0j}-p_{ij}^{2}\right)}}\left\{\int_{0}^{\frac{\pi}{2}} \frac{1}{\beta_{n}} d \theta \right.\\& \left. -\frac{\pi}{2N_{m}}\sum_{\varepsilon=1}^{N_{m}} D_{1}\left(\theta_{\varepsilon};p_{0i},p_{0j},p_{ij},d\right) \right.\\& \left. +\frac{\pi}{2N_{m}}\sum_{\varepsilon=1}^{N_{m}} D_{2}\left(\theta_{\varepsilon};p_{0i},p_{0j},p_{ij},d\right)\right\}-1,
\end{aligned}
\end{equation}
where $\theta_{\varepsilon}$ denotes the $\varepsilon$-th random number generated from the uniform distribution in the interval $\left[0,\frac{\pi}{2}\right]$. Note that the first part of the above integral was readily evaluated in closed-form in (\ref{eq:38}).

\subsection{Convex Covariance Recovery}
\label{subsec:12}
Similar the previous proposed approaches, we begin by estimating  $p_{0i}$ and $p_{0j}$ through (\ref{eq:fast_1}). We then aim at estimating the unknown parameters $\{p_{ij}\}$ by formulating a minimization problem. Namely, we consider the following criterion:
\begin{equation}
\label{eq:104}
\begin{aligned}
\bar{\Gamma}(p_{ij}) &\triangleq \log \left(\left|R_{\mathbf{y}}(i,j)-\chi\left\{\int_{0}^{\frac{\pi}{2}} \frac{1}{\beta_{n}} d \theta \right.\right.\right.\\& \left.\left.\left. -\frac{\pi}{2N_{m}}\sum_{\varepsilon=1}^{N_{m}} D_{1}\left(\theta_{\varepsilon};p_{0i}^{\star},p_{0j}^{\star},p_{ij},d\right)\right.\right.\right.\\& \left.\left.\left.+\frac{\pi}{2N_{m}}\sum_{\varepsilon=1}^{N_{m}} D_{2}\left(\theta_{\varepsilon};p_{0i}^{\star},p_{0j}^{\star},p_{ij},d\right)\right\}+1\right|^{2}\right),
\end{aligned}
\end{equation}
where the autocorrelation of output signal $R_{\mathbf{y}}$ can be estimated via (\ref{eq:43}). Let $F_{n}(p_{ij})$ denote the approximation of (\ref{eq:114}) using the Monte-Carlo integration. Therefore, we can consider the following alternative criterion:
\begin{equation}
\label{eq:105}
\Gamma(p_{ij}) \triangleq \log\left(\left|R_{\mathbf{y}}(i,j)-F_{n}(p_{ij})\right|^{2}\right).
\end{equation}
Similar to the previous criterion in (\ref{eq:96}), $\Gamma(p_{ij})$ appears to be a convex function with respect to $p_{ij}$, whose proof of convexity is similar to that for $\Psi_{m}(.)$ in Appendix~A---see Fig.~\ref{figure_15} for an example of the optimization landscape associated with $\Gamma(p_{ij})$. By considering the feasible region of the parameter of interest $\{p_{ij}\}$, one can formulate the following recovery problem:
\begin{equation}
\label{eq:106}
\begin{aligned}
\mathcal{P}_{i,j}&: &\min_{p_{ij}}& &\Gamma(p_{ij}),& &\mbox{s.t.}& &-p_{m}\leq p_{ij} \leq p_{m},
\end{aligned}
\end{equation}
where $p_{m}$ is defined in Section~\ref{subsec:4}. The convex problem in (\ref{eq:106}) may be tackled by the same tools as proposed in Section~\ref{subsec:8}. Recovery of $\{p_{ij}\}$ leads to estimating the autocorrelation values of $\mathbf{x}$ via (\ref{eq:36}). The optimum recovery results will be presented in the following.

\begin{figure}[t]
	\center{\includegraphics[width=0.6\textwidth]{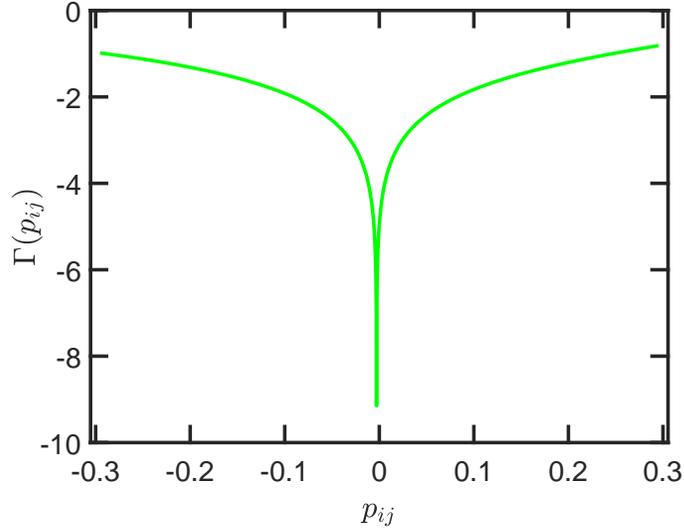}}
	\caption{Example plot of the estimation criterion $\Gamma(p_{ij})$ with respect to $p_{ij}$ showing its convexity.}
	\label{figure_15}
\end{figure}

\subsection{Numerical Results}
Herein, we examine the Monte-Carlo integration technique by comparing its recovery results with the true input signal autocorrelation values in the non-stationary case. In all experiments, the input signals were generated as zero-mean Gaussian sequences with time-varying variances ranging from $0.2$ to $0.8$. Accordingly, we made use of the time-varying thresholds with $d=0.3$ and diagonal $\bSigma$ whose diagonal entries are set to $0.1$.

We present an example of autocorrelation sequence recovery. The true input signal autocorrelation and the estimated autocorrelation values are shown in Fig.~\ref{figure_400}, with $i=2$ and $j$ being a temporal sequence of length $13$. It can be observed from Fig.~\ref{figure_400} that the Monto-Carlo based approach presents satisfactory recovery results in the non-stationary case as well.

\begin{figure}[t]
	\center{\includegraphics[width=0.6\textwidth]{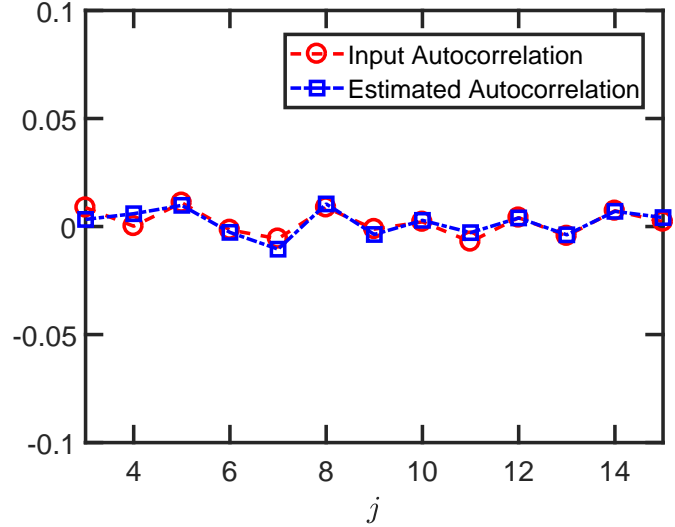}}
	\caption{Recovery of the input autocorrelation $r_{ij}$ using the Monte-Carlo integration approach, with $i=2$ and $j$ being a temporal sequence of length $13$. The true values are plotted alongside the estimates.}
	\label{figure_400}
\end{figure}

\section{Comparing The Proposed Recovery Methods}
It would be of interest to compare the discussed covariance recovery approaches in the non-stationary setting: (i) employing the \text{Padé} approximation of the integrands in (\ref{eq:114}), (ii) applying the Gauss-Legendre quadrature technique, and (iii) applying the Monte-Carlo integration to evaluate the integral in (\ref{eq:114}). To this end, we generate a non-stationary input signal $\mathbf{x}\in \mathbb{R}^{5}$ with the ensemble length $N_{\mathbf{x}}=10000$ and the following non-Toeplitz covariance matrix:
\begin{equation}
\label{eq:15oooo}
\begin{aligned}
\bR^{\star}_{\mathbf{x}}=
\left(\begin{array}{ccccc}
+0.5040 & -0.0065 & +0.0015 & -0.0036 & +0.0044 \\
-0.0065 & +0.2565 & -0.0034 & +0.0086 & +0.0031 \\
+0.0015 & -0.0034 & +0.3298 & +0.0063 & +0.0031 \\
-0.0036 & +0.0086 & +0.0063 & +0.6376 & -0.0062 \\
+0.0044 & +0.0031 & +0.0031 & -0.0062 & +0.4552
\end{array}\right).
\end{aligned}
\end{equation}
The time-varying threshold is generated using the same settings as in Section~\ref{subsec:22}. Table~\ref{table_1} illustrates the squared Frobenius norm of the error, normalized by the squared Frobenius norm of the desired covariance matrix $\bR^{\star}_{\mathbf{x}}$:
\begin{equation}
\label{eq:4000}
\mathrm{NMSE}\triangleq\frac{\left\|\bR^{\star}_{\mathbf{x}}-\hat{\bR}_{\mathbf{x}}\right\|^{2}_{\mathrm{F}}}{\left\|\bR^{\star}_{\mathbf{x}}\right\|^{2}_{\mathrm{F}}},
\end{equation}
where $\hat{\bR}_{\mathbf{x}}$ is the recovered covariance matrix. The presented results are averaged over $5$ experiments.

In the non-stationary input signal scenario, similar to the stationary case, all three approaches show promising recovery results---see Table~\ref{table_1}. The Gauss-Legendre method has a better performance in recovering the input signal autocorrelation values in comparison with the PA technique and the Monte-Carlo integration. It is also worth noting that the two proposed numerical approaches other than the PA technique boil down to simplified convex programs, hence ensuring convergence to the global optimum. However, a proper selection of the number of nodes and quadrature points in the Gauss-Legendre quadrature and the Monte-Carlo integration techniques is crucial and may present itself as a bottleneck in an effective recovery. This is not an obstacle in applying the PA technique. As a result, one may wish to run the PA-based recovery to help with the proper deployment of the other two techniques.

\begin{table} [t]
\caption{Average NMSE value for the Covariance Recovery.}
\centering
\begin{tabular}{ | c | c | }
\hline
\text { Covariance recovery approach } & \text { NMSE } \\ [0.5 ex]
\hline \hline
\text{PA technique} & $7.813e-05$ \\[1 ex]
\hline
\text{Gauss-Legendre quadrature} & $2.093e-05$ \\[1 ex]
\hline
\text{Monte-Carlo integration} & $2.488e-05$ \\[1 ex]
\hline
\end{tabular}
\label{table_1}
\end{table}
\section{Judicious Selection of Sampling Thresholds}
While the use of time-varying thresholds for one-bit sampling has shown promise in various signal recovery problems, tuning the applied thresholds provides both an opportunity and a challenge. In this section, we will discuss an approach to effectively set the sampling threshold mean $d$---whose significance was already shown in our analysis in Section~\ref{sec_fitness}. The value of the threshold mean $d$ is one of the parameters in our recovery cost functions, and consequently, impacts the effectiveness of the autocorrelation sequence recovery by various proposed approaches.

\subsection{Problem Formulation for Threshold Mean Optimization}

Consider a set of thresholds distributed as  $\btau\sim\mathcal{N}\left(\mathbf{d}=\mathbf{1}d,\bSigma=\sigma^{2}_{\btau} \bI\right)$. To design \emph{optimal} thresholds that are independent from the unknown zero-mean Gaussian signal ($\mathbf{x}$), we use the CDF of the observed sign data $\mathbf{y}$ and formulate a maximum likelihood estimation (MLE) problem. The goal is to determine the threshold mean $d$ solely from the sign data $\mathbf{y}$. The one-bit samples are generated as
\begin{equation}
\label{eq:1010}
\begin{aligned}
i\in\left\{1,\cdots,N\right\},\quad
y_{i} &= \begin{cases} +1 &  x_{i}>\tau_{i}, \\ -1 & x_{i}<\tau_{i}.
\end{cases}
\end{aligned}
\end{equation}
The probability vector $\bp$ for the one-bit measurement vector $\mathbf{y}$ may be written as
\begin{equation}
\label{eq:1020}
\begin{aligned}
p_{\mathbf{y}}\left(y_{i}|\tau_{i}\right)=p_{i} &= \begin{cases} 1-\Psi(\tau_{i}) & \text{for}\quad \{y_{i}=+1\}, \\ \Psi(\tau_{i}) &  \text{for}\quad \{y_{i}=-1\},
\end{cases}
\end{aligned}
\end{equation}
where $\Psi(.)$ is the CDF of $\mathbf{x}$, as defined in (\ref{eq:1bbb}). The associated log-likelihood function is hence
given by
\begin{equation}
\label{eq:1030}
\begin{aligned}
\mathcal{L}_{\mathbf{y}}(\br_{0},\btau) &= \sum^{N}_{i=1}\left\{\mathbb{I}_{(y_{i}=+1)}\log\left(1-\Psi(\tau_{i})\right) \right.\\& \left.+\mathbb{I}_{(y_{i}=-1)}\log\left(\Psi(\tau_{i})\right)\right\},
\end{aligned}
\end{equation}
where $\br_{0}$ is a vector containing the diagonal entries of the  covariance matrix of the input signal $\mathbf{x}$. As mentioned earlier, in the non-stationary scenario, the covariance matrix has an arbitrary non-Toeplitz structure. The entries of $\br_{0}$, appearing in the CDF, are the variances for the elements of the input signal. To immediately derive our desired parameter $d$ from (\ref{eq:1030}), we define the following statistical linear model for our threshold $\btau$:
\begin{equation}
\label{eq:1040}
\btau = \mathbf{1}d+\sigma^{2}_{\btau}~\mathbf{z}, \quad \mathbf{z}\sim\mathcal{N}\left(0,\bI\right).
\end{equation}
Therefore, the MLE is formulated as 
\begin{equation}
\label{eq:1050}
\begin{aligned}
\min_{d,\sigma^{2}_{\btau}}  \quad &-\mathcal{L}_{\mathbf{y}}(\br_{0},\btau)\\
\text{s.t.}\quad &r_{0i} = \left(\frac{d}{Q^{-1}\left(\frac{\mu_{i}+1}{2}\right)}\right)^{2}-\sigma^{2}_{\btau},\\
&i\in\left\{1,\cdots,N\right\},
\end{aligned}
\end{equation}
where the equality constraint is obtained from (\ref{rem1}).

\subsection{Numerical Illustrations for Threshold Mean Design}
To numerically scrutinize our approach, the input signal is generated using the same settings as described in Section~\ref{subsec:22}. The desired time-varying threshold is generated as a Gaussian process with $d^{\star}=0.3$ and $\Sigma=0.1\bI$. The sign data $\mathbf{y}$ was generated accordingly to be  utilized in order to estimate the desired threshold mean by the MLE problem in (\ref{eq:1050}). The results are presented in Fig.~\ref{figure_1050} based on the NMSE between the desired threshold mean $d^{\star}$ and the recovered mean $\hat{d}$, defined as:
\begin{equation}
\label{eq:260000}
\begin{aligned}
\mathrm{NMSE} \triangleq \frac{|d^{\star}-{\hat d}|^{2}}{|d^{\star}|^{2}}.
\end{aligned}
\end{equation}
Each presented data point is averaged over $5$ experiments. As can be seen in Fig.~\ref{figure_1050}, the proposed method can accurately estimate the mean of a time-varying threshold. The results are obtained with sequence lengths $N_{\mathbf{x}}\in \left\{1000, 3000, 6000, 10000\right\}$. Moreover, $\sigma^{2}_{\btau}$ is estimated using (\ref{eq:1050}) with the average NMSE of $\sim 1e-03$ by considering $N_{\mathbf{x}}=10000$.

\begin{figure}[t]
	\center{\includegraphics[width=0.45\textwidth]{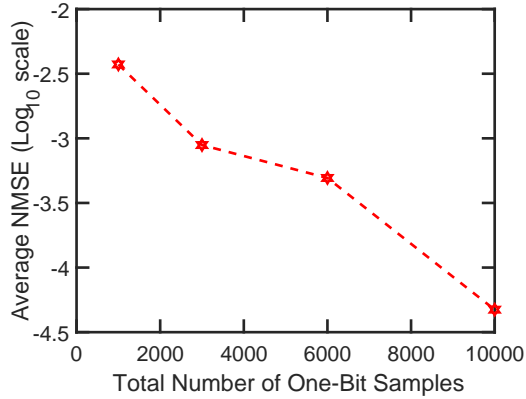}}
	\caption{Average NMSE for the estimated sampling threshold mean $\hat{d}$, based on MLE problem presented in (\ref{eq:1050}).}
	\label{figure_1050}
\end{figure}

\section{Modified Bussgang Law For Time-Varying Sampling Thresholds}
The modified Bussgang law for stationary input signals was derived in Part~\rom{1} of this work. This modified Bussgang law presents a useful relation in stochastic analysis of stationary input signals when they are sampled with time-varying thresholds. In this section, the modified Bussgang law is extended to the case when the non-stationary input signals are considered in such settings.
\begin{figure*}[t]
	\centering
	\begin{subfigure}[b]{0.45\textwidth}
		\includegraphics[width=1\linewidth]{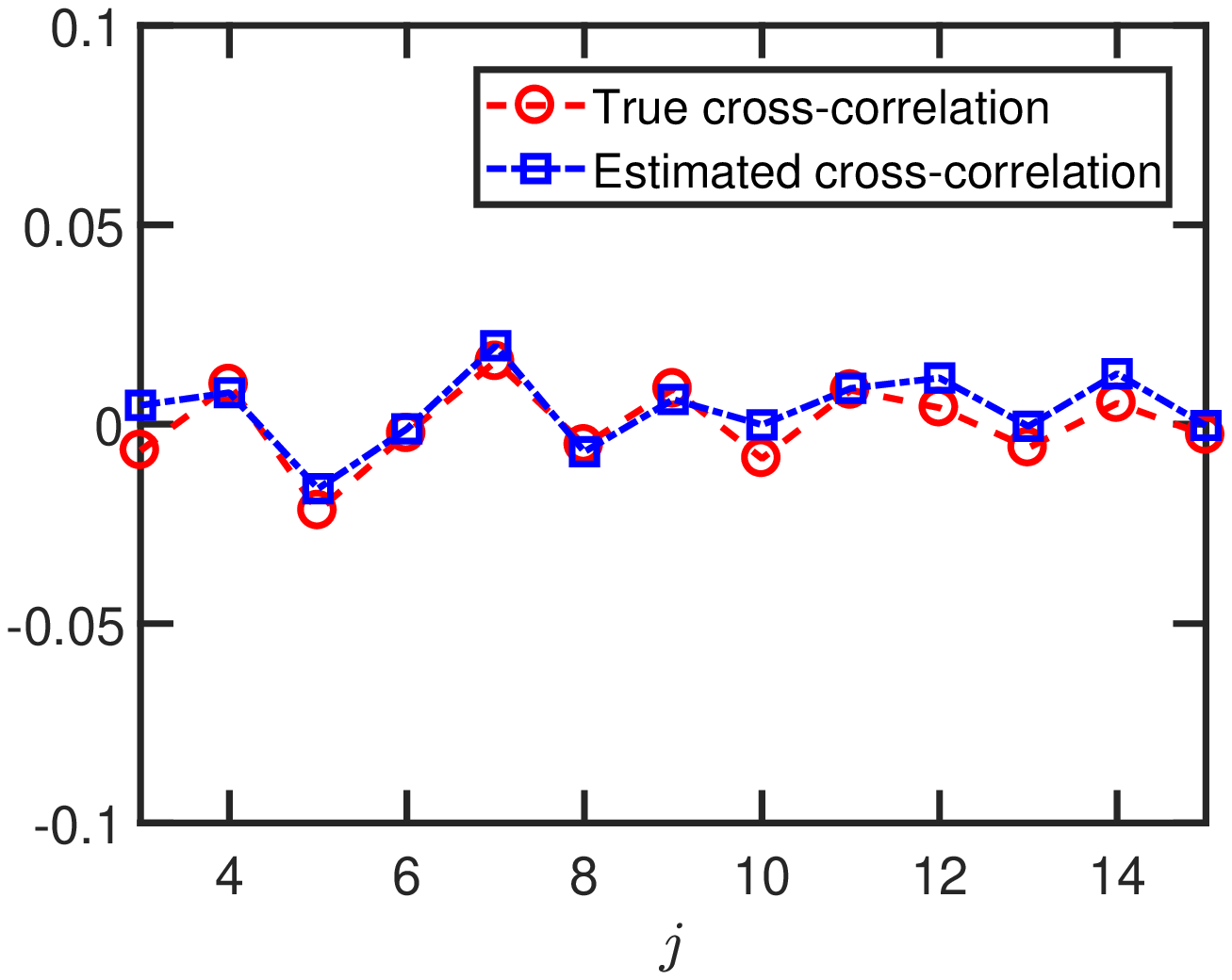}
		\caption{\text{Padé} approximation technique}
	\end{subfigure}
	\begin{subfigure}[b]{0.45\textwidth}
		\includegraphics[width=1\linewidth]{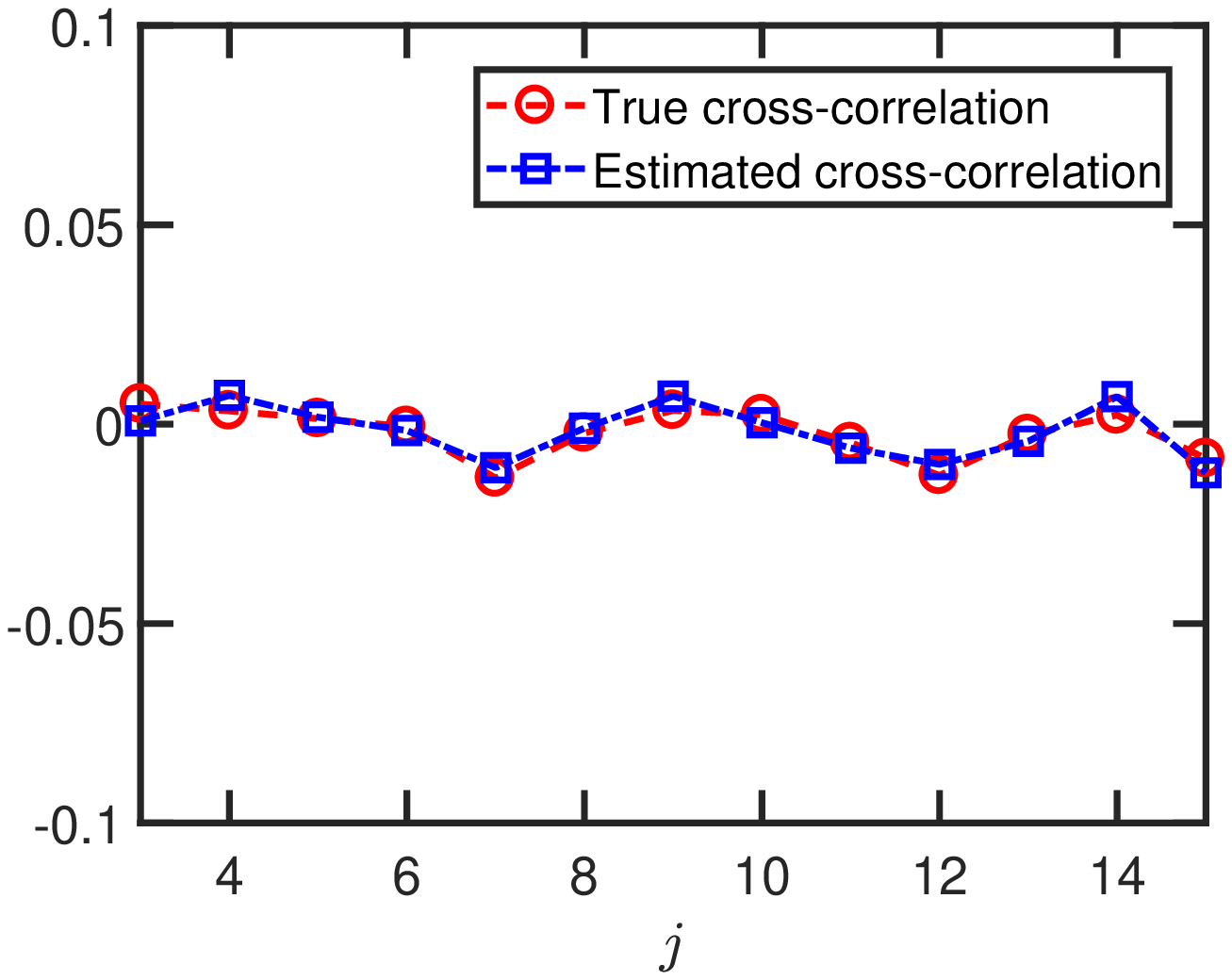}
		\caption{Gauss-Legendre quadrature technique}
	\end{subfigure}
	\begin{subfigure}[b]{0.45\textwidth}
		\includegraphics[width=1\linewidth]{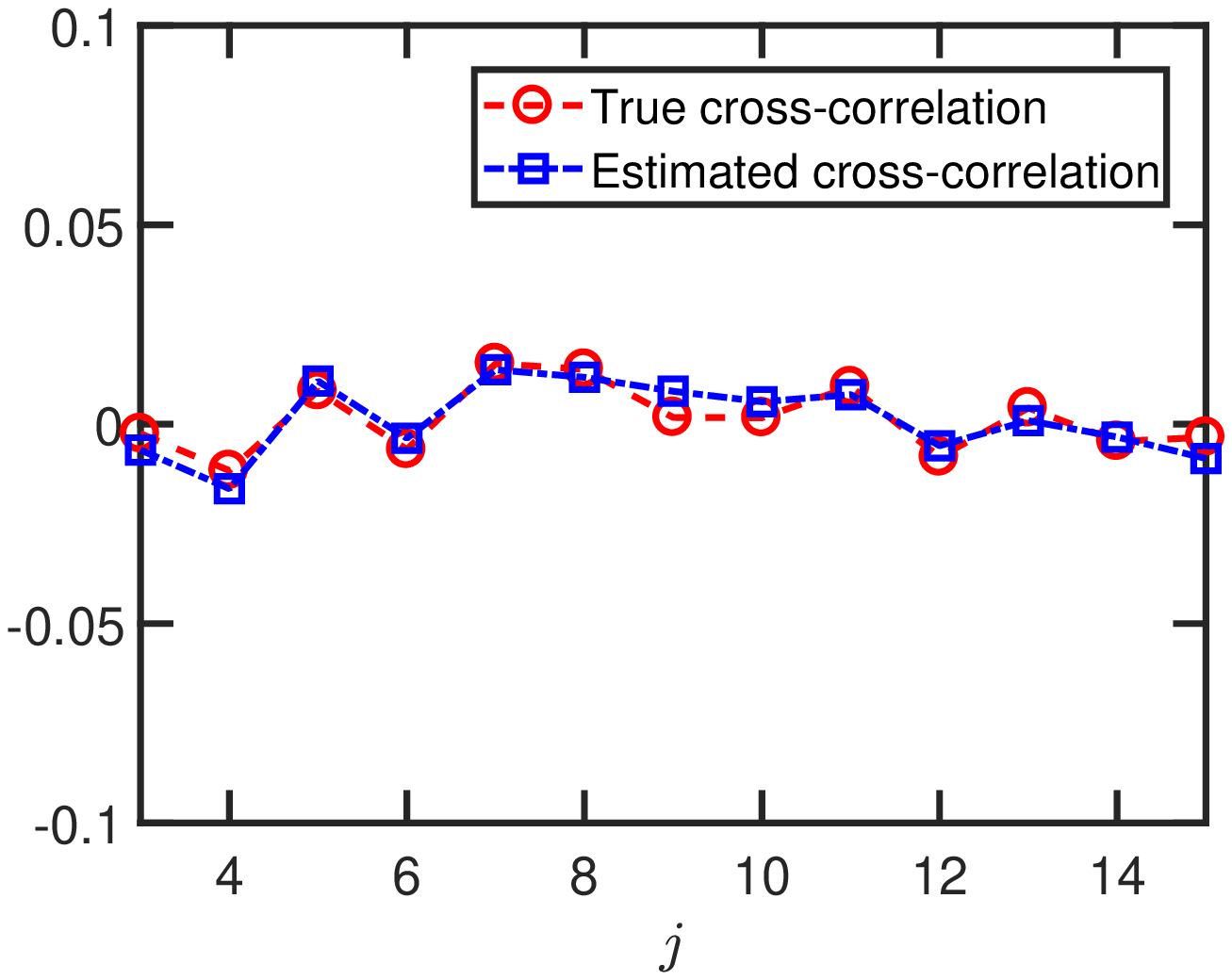}
		\caption{Monte-Carlo integration technique}
	\end{subfigure}
	\caption{The recovery of the cross-correlation between the input signal and the one-bit sampled data by the modified Bussgang law applied in conjunction with various one-bit autocorrelation recovery approaches for a sequence of length $13$, with the true values plotted alongside the estimates.}
	\label{figure_80}
\end{figure*}
\subsection{Modified Bussgang Law for Non-Stationary Input Signals}
By considering time-varying thresholds, the cross-correlation matrix between the one-bit sampled data and the non-stationary input signal can be formulated as follows.
\begin{theorem}
\label{theorem_4}
Suppose $\btau\sim\mathcal{N}\left(\mathbf{d}=\mathbf{1}d,\bSigma\right)$ is a time-varying threshold, and $\mathbf{x}$ is a non-stationary input signal. Let $\mathbf{y}=g\mathbf{(w)}$ denote the one-bit sampled data, where $\mathbf{w}=\mathbf{x}-\btau$ is distributed as $\mathbf{w}\sim \mathcal{N}\left(-\mathbf{d},\bR_{\mathbf{x}}+\bSigma=\bP\right)$, with $p_{0j}=\bP(j,j)$ and $p_{ij}=\bP(i,j)$. Then, the cross-correlation matrix between $\mathbf{y}$ and $\mathbf{x}$ satisfies the relation,
\begin{equation}
\label{eq:143}
\begin{aligned}
R_{\mathbf{yx}}(i,j) = R_{\mathbf{y}\btau}(i,j)+\left[\varepsilon_{1} p_{ij}-\varepsilon_{2} d(p_{0j}-p_{ij})\right],
\end{aligned}
\end{equation}
where $\varepsilon_{1}$ and $\varepsilon_{2}$ are given by
\begin{equation}
\label{eq:144}
\begin{aligned}
\varepsilon_{1} &= \sqrt{\frac{2}{\pi p_{0j}}}\Gamma\left(1,\dfrac{d^2}{2p_{0j}}\right)-\frac{d}{\sqrt{\pi p^{2}_{0j}}}\left(\Gamma\left(\dfrac{1}{2},\dfrac{d^2}{2p_{0j}}\right)-\sqrt{{\pi}}\right),\\ \varepsilon_{2} &= -\frac{1}{p_{0j}} \operatorname{erf}\left(\frac{d}{\sqrt{2p_{0j}}}\right).
\end{aligned}
\end{equation}
In particular, for $i=j$, the relation in (\ref{eq:143}) yields
\begin{equation}
\label{eq:108}
\begin{aligned}
R_{\mathbf{yx}}(i,i) &= R_{\mathbf{y}\btau}(i,i) + \sqrt{\frac{2p_{0i}}{\pi}}~
\Gamma\left(1,\frac{d^{2}}{2p_{0i}}\right)\\ &- \frac{d}{\sqrt{\pi}} \Gamma\left(\frac{1}{2},\frac{d^{2}}{2p_{0i}}\right) + d.
\end{aligned}
\end{equation}
\end{theorem}

\begin{proof}
Suppose $w_{i}$ and $w_{j}$ are the $i$-th and the $j$-th entries of $\mathbf{w}$ ($i\neq j$) with $\mathbb{E}\{w_{i}\}=\mathbb{E}\{w_{j}\}=-d$, and that $p_{0i}=\bP(i,i)$, $p_{0j}=\bP(j,j)$, and $p_{ij}=\bP(i,j)$, where $\bP$ denotes the covariance matrix of $\mathbf{w}$. Consider the quantized random variables $y_{i}=g(w_{i})$ and $y_{j}=g(w_{j})$, where $g(.)$ denotes a non-linear transformation function. Therefore, the cross-correlation function between $w_{i}$ and $y_{j}$ can be obtained as
\begin{equation}
\label{eq:206}
\begin{aligned}
R_{\mathbf{yw}}(i,j) &= \kappa \int^{\infty}_{-\infty}\int^{\infty}_{-\infty} w_{i}g(w_{j})e^{\lambda(d)}dw_{i}dw_{j},
\end{aligned}
\end{equation}
where $\kappa$ and $\lambda(d)$ are defined in (\ref{eq:86}) and (\ref{eq:28}), respectively. We first calculate the integral in (\ref{eq:206}) with respect to $w_{i}$, i.e.,
\begin{equation}
\label{eq:88}
\begin{aligned}
R_{\mathbf{yw}}(i,j) &= \frac{e^{\frac{-d^{2}(p_{0i}+p_{0j}-2p_{ij})}{2(p_{0i}p_{0j}-p_{ij}^{2})}}}{2\pi\sqrt{p_{0i}p_{0j}-p_{ij}^2}} \int^{\infty}_{-\infty}  g(w_{j})e^{\frac{2d(p_{0i}-p_{ij})w_{j}+w^{2}_{j}p_{0i}}{-2(p_{0i}p_{0j}-p_{ij}^{2})}}\\
&\int^{\infty}_{-\infty} w_{i}e^{\frac{2d(p_{0j}-p_{ij})w_{i}+w^{2}_{i}p_{0j}-2p_{ij}w_{i}w_{j}}{-2(p_{0i}p_{0j}-p_{ij}^{2})}}dw_{i}dw_{j}\\
&= \left[ \varepsilon_{1}p_{ij}-\varepsilon_{2}d(p_{0j}-p_{ij})\right],
\end{aligned}
\end{equation}
where $\varepsilon_{1}$ and $\varepsilon_{2}$ are given as
\begin{equation}
 \label{eq:90}
 \begin{aligned}
 \varepsilon_{1}(j) &= \frac{1}{\sqrt{2\pi p^{3}_{0j}}}\int^{\infty}_{-\infty} w_{j}g(w_{j})e^{\frac{-(w_{j}+d)^{2}}{2p_{0j}}}d w_{j},\\
 \varepsilon_{2}(j) &= \frac{1}{\sqrt{2\pi p^{3}_{0j}}}\int^{\infty}_{-\infty}g(w_{j})e^{\frac{-(w_{j}+d)^{2}}{2p_{0j}}}d w_{j}.
 \end{aligned}
\end{equation}
A detailed proof of the results in (\ref{eq:88}) and (\ref{eq:90}) is presented in Appendix~B. Based on (\ref{eq:90}), it can be seen that the values of $\varepsilon_{1}$ and $\varepsilon_{2}$ are dependent on the entry index number $j$. As a result, the modified Bussgang law for the non-stationary input signal can be presented as:
\begin{equation}
\label{eq:92}
\begin{aligned}
R_{\mathbf{yx}}(i,j)-R_{\mathbf{y}\btau}(i,j)=& \left[(\varepsilon_{1}(j)+d \varepsilon_{2}(j))\right.\\& \left.\left(R_{\mathbf{x}}(i,j)+\bSigma(i,j)\right)-d \varepsilon_{2}(j)p_{0j}\right].
\end{aligned}
\end{equation}
If the non-linear function $g(.)$ is the sign function, which is the case in one-bit quantization, $\varepsilon_{1}$ and $\varepsilon_{2}$ can be obtained using similar steps as presented in Part~\rom{1} of our work. However, in this case, $p_{0j}$ should be utilized in lieu of $p_{0}$. For $i=j$,  Eq.~(\ref{eq:92}) boils down to (\ref{eq:108}), a proof of which is presented in Appendix~C. Based on (\ref{eq:92}), the cross-correlation matrix between the input and the output one-bit data is computed, where $\{p_{0j}\}$ is obtained by (\ref{eq:fast_1}) and $\{p_{ij}\}$ can be either recovered using (\ref{eq:45}), (\ref{eq:97}) or (\ref{eq:106}). Note that the cross-correlation matrix between the threshold vector $\btau$ and the output vector $\mathbf{y}$ can be estimated via a \emph{sample cross-correlation matrix}:
\begin{equation}
\label{eq:85}
\bR_{\mathbf{y}\btau}\approx \frac{1}{N_{\mathbf{x}}} \sum_{k=1}^{N_{\mathbf{x}}} \mathbf{y}(k) \btau(k)^{\mathrm{H}}.
\end{equation}
\end{proof}

\subsection{A Numerical Investigation of the Modified Bussgang Law}
We now examine the modified Bussgang law for the non-stationary input signals by comparing its recovery results with the true cross-correlation values between the input signal and one-bit quantized data. In all experiments, the input signal settings are the same as Section~\ref{subsec:22}. The time-varying threshold settings are as follows: (a) PA: $d=0.5$ and $\bSigma=0.2\bI$, (b) Gauss-Legendre: $d=0.3$ and $\bSigma=0.1\bI$, (c) Monte-Carlo: $d=0.3$
and $\bSigma=0.1\bI$ , where $\bI$ denotes the identity matrix. 

The true cross-correlation between the input signal and the one-bit sampled data and the estimated cross-correlation values obtained using our approach are shown in Fig.~\ref{figure_80}, for $i=2$ and $j$ as a random sequence of length $13$. Our results appear to confirm the possibility of recovering the cross-correlation values from one-bit sampled data with time-varying thresholds by employing any of the three recovery methods (PA, Gauss-Legendre method and Monte-Carlo integration).

\section{Conclusion}
We studied a generalization of the modified arcsine law discussed in Part~\rom{1} of our work through \text{Padé} approximations, Gauss-Legendre quadrature approach, and Monte-Carlo integration, to cases where the input signal is assumed to be non-stationary. The numerical results present the efficacy of all three approaches in the covariance matrix recovery. Moreover, a modified Bussgang law was established for the one-bit sampling of non-stationary input signals with time-varying thresholds.

\appendices

\section{Proof of The Convexity of $\Psi(p_{ij})$ in (\ref{eq:96})}
Since $\log(\cdot)$ is a strictly increasing function, we can analyze the criterion $\Psi_{m}(p_{ij})=\left(R_{\mathbf{y}}(i,j)-J_{n}(p_{ij})\right)^{2}$ to show the convexity of $\Psi(p_{ij})$. The derivative of $\Psi_{m}(p_{ij})$ with respect to $p_{ij}$ is computed as
\begin{equation}
\label{eq:app_1}
\Psi_{m}^{\prime}(p_{ij})=-2\left(R_{\mathbf{y}}(i,j)-J_{n}(p_{ij})\right)J_{n}^{\prime}(p_{ij}),
\end{equation}
where $J_{n}$ is the approximated version of (\ref{eq:114}) using the Gauss-Legendre quadrature with the following close-form formula:
\begin{equation}
\label{eq:app_2}
\begin{aligned}
J_{n}(p_{ij})=\chi(p_{ij})\Bigg(\sqrt{p_{0i}p_{0j}-p_{ij}^{2}}&\left(\pi+2\sin^{-1}\left(\frac{p_{ij}}{\sqrt{p_{0i}p_{0j}}}\right)\right)\\&\frac{\pi}{4} I\Bigg)-1=\chi(p_{ij})T(p_{ij})-1,
\end{aligned}
\end{equation}
where $\chi$ is defined in (\ref{eq:79}), and $I$ is given by
\begin{equation}
\label{eq:app_3}
I = \sum_{\varepsilon=1}^{N_{q}}\omega_{\varepsilon}\sqrt{\frac{\pi}{\beta_{n}}}\left(\frac{\alpha_{n}}{\beta_{n}}\right)\left(\frac{1}{2}-Q\left(\frac{\alpha_{n}}{\sqrt{2\beta_{n}}}\right)\right)e^{\frac{\alpha_{n}^{2}}{4\beta_{n}}}.
\end{equation}
Based on (\ref{eq:app_2}) and (\ref{eq:app_3}), $J_{n}^{\prime}(p_{ij})$ can be formulated as
\begin{equation}
\label{eq:app_4}
\begin{aligned}
J_{n}^{\prime}(p_{ij})&=\chi(p_{ij})\left(2-\frac{p_{ij}\left(\pi+2\sin^{-1}\left(\frac{p_{ij}}{\sqrt{p_{0i}p_{0j}}}\right)\right)}{\sqrt{p_{0i}p_{0j}-p_{ij}^{2}}}\right)\\&+\chi(p_{ij})\left(\frac{\partial I}{\partial \alpha_{n}}\frac{\partial \alpha_{n}}{\partial p_{ij}}+\frac{\partial I}{\partial \beta_{n}}\frac{\partial \beta_{n}}{\partial p_{ij}}\right)\\&+\frac{\partial \chi}{\partial p_{ij}}T(p_{ij}),
\end{aligned}
\end{equation}
where $\frac{\partial \alpha_{n}}{\partial p_{ij}}$ and $\frac{\partial \beta_{n}}{\partial p_{ij}}$ are given according to (\ref{eq:115}). As can be seen in (\ref{eq:app_1}), (\ref{eq:app_2}), and (\ref{eq:app_4}), analyzing the convexity of $\Psi_{m}(p_{ij})$ depends on the parameters $d$, $p_{0i}$, $p_{0j}$, $N_{q}$, and $\{\theta_{\varepsilon}\}$, which indicates that the analysis is restricted to the case where the mentioned parameters are known; i.e. the parameters must be specified for the covariance matrix recovery. Generally speaking, based on (\ref{eq:app_1}), (\ref{eq:app_2}) and (\ref{eq:app_4}), $\Psi_{m}(p_{ij})$ is convex when $J_{n}^{\prime}(p_{ij})>0$, or equivalently when $J_{n}(p_{ij})$ is a strictly increasing function in the feasible region of $p_{ij}$; i.e. $-p_{m} \leq p_{ij} \leq p_{m}$ where $p_{m}=\operatorname{min}\{[p_{0i},p_{0j}]\}$. As a result, $\Psi_{m}^{\prime}(p_{ij})=0$ has only one solution which is the value of $p_{ij}$ that satisfies $R_{\mathbf{y}}(i,j)=J_{n}(p_{ij})$. Therefore, the convexity of $\Psi_{m}(p_{ij})$ can be easily concluded in light of (\ref{eq:app_1}). For instance, one may easily verify that the selected parameters for the recovery of the input covariance matrix in Section~\ref{subsec:30} makes $J_{n}(p_{ij})$ a strictly increasing function, and thus, $\Psi_{m}(p_{ij})$ a convex function.

\section{Proof of The Modified Bussgang Law Formula}
Note that
\begin{equation}
\label{eq:299}
\begin{aligned}
R_{\mathbf{yw}}(i,j) &= \frac{e^{\frac{-d^{2}(p_{0i}+p_{0j}-2p_{ij})}{2(p_{0i}p_{0j}-p_{ij}^{2})}}}{2\pi\sqrt{p_{0i}p_{0j}-p_{ij}^2}} \int^{\infty}_{-\infty}  g(w_{j})e^{\frac{2d(p_{0i}-p_{ij})w_{j}+w^{2}_{j}p_{0i}}{-2(p_{0i}p_{0j}-p_{ij}^{2})}}\\
&\int^{\infty}_{-\infty} w_{i}e^{\frac{2d(p_{0j}-p_{ij})w_{i}+w^{2}_{i}p_{0j}-2p_{ij}w_{i}w_{j}}{-2(p_{0i}p_{0j}-p_{ij}^{2})}}dw_{i}dw_{j},
\end{aligned}
\end{equation}
where the inner integral and the outer integral are called $\mathcal{L}_{1}$ and $\mathcal{L}_{2}$, respectively. The inner integral may be evaluated as
\begin{equation}
\label{eq:300}
\begin{aligned}
\mathcal{L}_{1}&= \int^{\infty}_{-\infty}
w_{i}e^{\frac{2d(p_{0j}-p_{ij})w_{i}+w^{2}_{i}p_{0j}-2p_{ij}w_{i}w_{j}}{-2(p_{0i}p_{0j}-p^{2}_{ij})}}dw_{i}\\
&= e^{\frac{\left(p_{0j}d-p_{ij}\left(w_{j}+d\right)\right)^2}{2p_{0j}(p_{0i}p_{0j}-p^{2}_{ij})}}\int^{\infty}_{-\infty}
w_{i}e^{-\frac{\left(w_{i}+\left(d-\frac{p_{ij}}{p_{0j}}\left(w_{j}+d\right)\right)\right)^2}{2(p_{0j}-\frac{p^{2}_{ij}}{p_{0j}})}}dw_{i}\\
&= e^{\frac{\left(p_{0j}d-p_{ij}\left(w_{j}+d\right)\right)^2}{2p_{0j}(p_{0i}p_{0j}-p^{2}_{ij})}} \sqrt{2\pi \left(p_{0j}-\frac{p^{2}_{ij}}{p_{0j}}\right)}\times \cdots\\
&\left(\frac{p_{ij}}{p_{0j}}\left(w_{j}+d\right)-d\right).
\end{aligned}
\end{equation}
Next, the outer integral is evaluated as
\begin{equation}
\label{eq:301}
\begin{aligned}
\mathcal{L}_{2} &= \sqrt{2\pi \left(p_{0j}-\frac{p^{2}_{ij}}{p_{0j}}\right)} e^{\frac{p^{2}_{0j}d^{2}+p^{2}_{ij}d^{2}-2dp_{0j}p_{ij}}{2p_{0j}(p_{0i}p_{0j}-p^{2}_{ij})}} \times \cdots \\ &\int^{\infty}_{-\infty} g(w_{j}) \left(\frac{p_{ij}}{p_{0j}}\left(w_{j}+d\right)-d\right) e^{\frac{\alpha_{ij}\left(w^{2}_{j}+2dw_{j}\right)}{\beta_{ij}}} dw_{j},
\end{aligned}
\end{equation}
where $\alpha_{ij}=p_{0i}p_{0j}-p^{2}_{ij}$ and $\beta_{ij}=-2p_{0j}\alpha_{ij}$. The above integral can be simplified as below:
\begin{equation}
\label{eq:302}
\begin{aligned}
\mathcal{L}_{2} &= \sqrt{2\pi \left((p_{0j}-\frac{p^{2}_{ij}}{p_{0j}})\right)} e^{\frac{d^{2}(p_{0i}+p_{0j}-2p_{ij})}{2(p_{0i}p_{0j}-p_{ij}^{2})}} \times \cdots\\ &\int^{\infty}_{-\infty}g(w_{j})\left(\frac{p_{ij}}{p_{0j}}\left(w_{j}+d\right)-d\right)e^{-\frac{\left(w_{j}+d\right)^2}{2p_{0j}}} dw_{j}.
\end{aligned}
\end{equation}
When the values from (\ref{eq:300}) and (\ref{eq:302}) are inserted in Eq. (\ref{eq:299}), we have the final form of the modified Bussgang law, i.e.
\begin{equation}
\label{eq:303}
\begin{aligned}
R_{\mathbf{yw}}(i,j)=\varepsilon_{1}p_{ij}-\varepsilon_{2} d(p_{0j}-p_{ij}),
\end{aligned}
\end{equation}
where $\varepsilon_{1}$ and $\varepsilon_{2}$ are obtained as below:
\begin{equation}
\label{eq:304}
\begin{aligned}
\varepsilon_{1} &= \frac{1}{\sqrt{2\pi p^{3}_{0j}}}\int^{\infty}_{-\infty} w_{j}g(w_{j})e^{-\frac{(w_{j}+d)^{2}}{2p_{0j}}}d w_{j},\\ \varepsilon_{2} &= \frac{1}{\sqrt{2\pi p^{3}_{0j}}}\int^{\infty}_{-\infty}g(w_{j})e^{-\frac{(w_{j}+d)^{2}}{2p_{0j}}}d w_{j}.
\end{aligned}
\end{equation}

\section{Proof of Equation (\ref{eq:108})}
To prove (\ref{eq:108}), we may formulate $\mathbb{E}\{w_{i}y_{i}\}$ as
\begin{equation}
\label{eq:109}
\begin{aligned}
\mathbb{E}\{w_{i}y_{i}\} &= \int_{-\infty}^{\infty} w_{i}f(w_{i})p(w_{i}) \,dw_{i}\\ &= \int_{0}^{\infty} w_{i}p(w_{i}) \,dw_{i} - \int_{-\infty}^{0} w_{i}p(w_{i}) \,dw_{i}\\ &= \int_{0}^{\infty} w_{i} \left(p(w_{i})+p(-w_{i})\right) \,dw_{i},
\end{aligned}
\end{equation}
where $p(w_{i})=\left(\sqrt{2\pi p_{0i}}\right)^{-1} e^{\frac{-\left(w_{i}+d\right)^{2}}{2p_{0i}}}$. Let $I_{1}=\int_{0}^{\infty}w_{i}e^{\frac{-\left(w_{i}+d\right)^{2}}{2p_{0i}}} \,dw_{i}$ and $I_{2}=\int_{0}^{\infty}w_{i}e^{\frac{-\left(-w_{i}+d\right)^{2}}{2p_{0i}}} \,dw_{i}$. Then, we have
\begin{equation}
\label{eq:109n}
\begin{aligned}
I_{1} &=\int_{0}^{\infty}w_{i}e^{\frac{-\left(w_{i}+d\right)^{2}}{2p_{0i}}} \,dw_{i} = \int_{0}^{\infty}\left(w_{i}+d\right)e^{\frac{-\left(w_{i}+d\right)^{2}}{2p_{0i}}} \,dw_{i}\\&-d\int_{0}^{\infty}e^{\frac{-\left(w_{i}+d\right)^{2}}{2p_{0i}}} \,dw_{i} = d\sqrt{\frac{\pi p_{0i}}{2}}\left(-1+\operatorname{erf}\left(\frac{d}{\sqrt{2p_{0i}}}\right)\right)\\&+p_{0i}e^{-\frac{d^{2}}{2p_{0i}}},\\I_{2}&=\int_{0}^{\infty}w_{i}e^{\frac{-\left(-w_{i}+d\right)^{2}}{2p_{0i}}} \,dw_{i} = \int_{0}^{\infty}\left(w_{i}-d\right)e^{\frac{-\left(-w_{i}+d\right)^{2}}{2p_{0i}}} \,dw_{i}\\&+d\int_{0}^{\infty}e^{\frac{-\left(-w_{i}+d\right)^{2}}{2p_{0i}}} \,dw_{i} = d\sqrt{\frac{\pi p_{0i}}{2}}\left(1+\operatorname{erf}\left(\frac{d}{\sqrt{2p_{0i}}}\right)\right)\\&+p_{0i}e^{-\frac{d^{2}}{2p_{0i}}}.
\end{aligned}
\end{equation}
Therefore, based on (\ref{eq:109n}), a simplified form  of the integration in (\ref{eq:109}) can be proposed as
\begin{equation}
\label{eq:109nn}
\mathbb{E}\{w_{i}y_{i}\} = d\operatorname{erf}\left(\frac{d}{\sqrt{2p_{0i}}}\right)+\sqrt{\frac{2p_{0i}}{\pi}}e^{-\frac{d^{2}}{2p_{0i}}}.
\end{equation}
Since $\Gamma\left(1,x\right)=e^{-x}$ and $\Gamma\left(\frac{1}{2},x\right)=\sqrt{\pi}\left(1-\operatorname{erf}\left(\sqrt{x}\right)\right)$, (\ref{eq:109nn}) is identical to (\ref{eq:108}).
\bibliographystyle{IEEEbib}
\bibliography{strings,refs}

\end{document}